\documentclass{article}

\usepackage{graphicx}

\usepackage{authblk}

\usepackage{soul}

\usepackage{booktabs}

\usepackage{amsthm}

\usepackage[
    backend=biber,
    style=numeric,
    sorting=anyt
]{biblatex}
\usepackage{amssymb}  
\usepackage{amsmath}  

\newcommand{\inb}[1]{\left(#1\right)}

\newcommand{\normal}[2]{\mathcal{N}(#1, #2)}

\newcommand{\fig}[1]{Fig.~\ref{#1}}
\newcommand{\tab}[1]{Tab.~\ref{#1}}
\newcommand{\secn}[1]{Sec.~\ref{#1}}

\newcommand{\std}{\sigma}

\newcommand{\vectorized}[1]{\tilde{#1}}

\newcommand{\feature}{x}
\newcommand{\design}{\vectorized{\feature}}
\newcommand{\parameter}{\theta}
\newcommand{\target}{y}
\newcommand{\observations}{\vectorized{\target}}

\newcommand{\designsize}{n}
\newcommand{\mean}{\mu}
\newcommand{\variance}{V}
\newcommand{\covariance}{\Sigma}
\newcommand{\model}{f}

\newcommand{\featurespace}{\mathcal{X}}
\newcommand{\parameterspace}{\Theta}
\newcommand{\targetspace}{\mathcal{Y}}

\newcommand{\noise}{\varepsilon}

\newcommand{\perturbation}{z}  
\newcommand{\Jacobian}{J}
\newcommand{\IM}{M}

\newcommand{\unitvector}{e}
\newcommand{\unitmatrix}{I}

\newcommand{\Realspace}{\mathbb{R}}
\newcommand{\estimator}[1]{\hat{#1}}
\newcommand{\true}[1]{#1^*}

\newcommand{\argmin}{\text{argmin}}
\newcommand{\varianceoperator}[1]{\text{Var}\left[#1\right]}
\newcommand{\covarianceoperator}[1]{\text{Cov}\left[#1\right]}
\newcommand{\meanoperator}[1]{\text{E}\left[#1\right]}
\newcommand{\meanof}[1]{\bar{#1}}
\newcommand{\probabilityof}[1]{p\inb{#1}}
\newcommand{\numberof}[1]{N_{#1}}
\newcommand{\dimensionof}[1]{d_{#1}}

\newcommand{\lin}{\text{LIN}}

\newcommand{\pd}{\mathrm{pd}}
\newcommand{\sigmapoint}{\text{SP}}
\newcommand{\cubature}{\text{CUB}}

\newcommand{\mcnameestenger}{\text{MS}}
\newcommand{\ludarmofal}{\text{LD}}
\newcommand{\montecarlo}{\text{MC}}

\newcommand{\leastsquares}{\text{least-squares}~}

\usepackage[colorinlistoftodos,prependcaption,textsize=tiny]{todonotes}

\newtheorem{thm}{Theorem}
\newtheorem{cor}[thm]{Corollary}
\newtheorem{prop}[thm]{Proposition}
\newtheorem{lm}[thm]{Lemma}

\newcommand{\ol}{\overline}

\renewcommand{\d}{\,\mathrm{d}} 
\newcommand{\N}{\mathbb{N}}
\newcommand{\diag}{\operatorname{diag}}
\newcommand{\E}{\operatorname{E}}
\newcommand{\norm}[1]{\left\| #1 \right\|}
\newcommand{\normtwo}[1]{\norm{#1}_2}
\newcommand{\R}{\mathbb{R}}
\newcommand{\dimx}{d_{\feature}}
\newcommand{\dimp}{d_{\parameter}}

\title{Cubature-based uncertainty estimation for nonlinear regression models}

\author[1]{Martin Bubel}
\author[1]{Jochen Schmid}
\author[1]{Maximilian Carmesin}
\author[2]{Volodymyr Kozachynskyi}
\author[2]{Erik Esche}
\author[1]{Michael Bortz}
\affil[1]{Department of Optimization, Fraunhofer Institute for Industrial Mathematics, Fraunhofer-Platz 1, D-67663 Kaiserslautern, Germany}
\affil[2]{Process Dynamics and Operations Group, Technische Universität Berlin, Straße des 17. Juni 135, D-10623 Berlin, Germany}

\date{}

\begin{document}

\maketitle

\begin{abstract}
    Calibrating model parameters to measured data by minimizing loss functions is an important step in obtaining realistic predictions from model-based approaches, e.g., for process optimization. This is applicable to both knowledge-driven and data-driven model setups. Due to measurement errors, the calibrated model parameters also carry uncertainty. In this contribution, we use cubature formulas based on sparse grids to calculate the variance of the regression results. The number of cubature points is close to the theoretical minimum required for a given level of exactness. We present exact benchmark results, which we also compare to other cubatures. This scheme is then applied to estimate the prediction uncertainty of the NRTL model, calibrated to observations from different experimental designs.
\end{abstract}

\section{Introduction} \label{sec:introduction}
In chemical engineering and many other fields, models are used to describe real-world processes or phenomena.
A model is a function $\model$ that describes how some target quantity $\target \in \targetspace \subset \Realspace^{\dimensionof{\target}}$ of interest depends on relevant input quantities $x \in \featurespace \subset \Realspace^{\dimensionof{\feature}}$ and model parameters $\parameter \in \parameterspace \subset \Realspace^{\dimensionof{\parameter}}$.
Typically, the model parameters $\parameter$ are calibrated to observed data $\observations=\inb{\target_1, \target_2, \dots, \target_{\designsize}}$. And these observations are obtained by performing experiments or running simulations at design $\design = \inb{\feature_1, \feature_2, \dots, \feature_{\designsize}}$, where $\designsize$ is the number of experiments.
In real-world processes, the target values are observed under random additive noise, which is usually assumed to be zero-mean normally distributed
\begin{equation} \label{eq:noisymodel}
    \target_i = \model\inb{\feature_i, \true{\parameter}} + \noise_i, \quad \noise_i \sim \normal{0}{\covariance} ,
\end{equation}
where $\noise_i$ is the measurement error of the $i$th observation, $\true{\parameter}$ is the unknown true parameter value, and $\covariance$ is the covariance matrix of observation noise on the target. This matrix might be known from device information or simulation specifications, or it is estimated from the available observation data. The assumption of the existence of a fixed set of (true but unknown) model parameters $\true{\parameter}$ and random observations is known as the frequentist view of modeling \cite{Bishop2009}. In this modeling view, which is predominant in many engineering disciplines, parameter estimates are obtained by minimizing some loss function, typically the sum of squared errors defined as
\begin{equation} \label{eq:leastsquares}
    S_{\design,\observations}\inb{\parameter} := \sum_{i=1}^{\designsize} \inb{\target_i - \model\inb{\feature_i, \parameter}}^2.
\end{equation}
In this manner, one obtains a \leastsquares parameter estimator 
\begin{equation} \label{eq:parameterestimate}
    \estimator{\parameter}_{\model}\inb{\design, \observations} \in \argmin_{\parameter \in \parameterspace} S_{\design,\observations}\inb{\parameter}
\end{equation}
for $\model$ based on the observation data $(\design,\observations)$. In general, there exist several \leastsquares estimators for a given set of observation data, but this does not affect the methodology introduced here.

As we observe the targets under random noise \eqref{eq:noisymodel}, the obtained parameter estimate $\estimator{\parameter}_{\model}\inb{\design, \observations} \in \parameterspace$ is also a random variable. In other words, if we repeat the experiments in $\design$, we will collect different observations $\observations$ and thus obtain different estimates $\estimator{\parameter}_{\model}\inb{\design, \observations}$. Accordingly, the model predictions 
\begin{align} \label{eq:modelpredictions}
\model\inb{\feature, \estimator{\parameter}_{\model}\inb{\design, \observations}}
\end{align}
obtained from those parameter estimators will differ as well, resulting in a distribution of predictions. Important indicators of this predictive distribution are its mean and its variance which are also known, respectively, as the expected prediction and the prediction uncertainty. 

In the following, we consider two types of models: linear and nonlinear ones. As usual, a (non)linear model is a model that is (non)linear in the parameters $\parameter$ (the dependence on the inputs $\feature$ being irrelevant). It is well-known that for linear models, there are closed-form expressions for the expected prediction and the prediction uncertainty. In contrast, for nonlinear models, there usually are no such closed-form expressions. In particular, the distribution of the parameter estimate $\estimator{\parameter}_{\model}\inb{\design, \observations}$ is generally unknown. As a consequence, the distribution of model predictions~\eqref{eq:modelpredictions} is even less known, by the nonlinearity of the map $\parameter \mapsto \model\inb{\feature, \parameter}$. This necessitates the approximation of the predictive distribution -- and, in particular, of the prediction uncertainty -- for nonlinear regression models.

In the literature, there are many methods to approximate the prediction uncertainty of nonlinear regression models. The most common approach is probably linearization, see \cite{Fahrmeir2009,Bates1988}. However, as the name indicates, linearization is quite a big simplification which generally ist not very accurate for highly nonlinear models \cite{Seber2003}. A straightforward improvement would be to consider higher-order terms in the Taylor series expansion of the model around the parameter estimate, as done in \cite{Bates1988}. However, that would require the computation of higher-order derivatives of the model, which is often infeasible in practice.
There are other approaches such as Monte Carlo sampling, which are mainly applicable in settings where experiments can be repeated precisely, e.g., in simulation-based settings, see e.g. Kozachynskyi et al. \cite{Kozachynskyi2024}. If experiments are not repeated, as in Krausch et al. \cite{Krausch2019}, the obtained uncertainty estimator may be biased.
In the field of statistics and machine learning, there are also methods such as bootstrapping \cite{Efron1994} and conformal prediction \cite{Vovk2005}. However, those methods often do not perform well in the low-data regime, which is predominant in chemical engineering where datasets are typically small. This is why we focus on an alternative approximation approach, which is based on cubature methods with small numbers of cubature points.

In this work, we review existing estimation approaches of the prediction uncertainty of nonlinear regression models (\secn{sec:montecarloapproximation} and \ref{sec:linearapproximation}) and propose new methods based on cubature formulas (\secn{sec:cubatureapproximation}). In \secn{sec:validationandbenchmarking}, we validate our approach on a generic quadratic model, for which we derive exact expressions of the prediction uncertainty for special designs. We further perform benchmarks of the proposed uncertainty approximation methods on a range of regression case studies. We compare the performance of the proposed methods and discuss the influence of the experimental design on the prediction uncertainty of a calibrated model. The benchmarking models include three toy models and one model from chemical engineering. For two of the toy models, we contribute exact expressions for the parameter estimator and for the prediction uncertainty, under the condition of a factorial design. In \secn{sec:conclusionandoutlook}, we summarize the key findings of this work and give an outlook on future research.

\section{Approximating the prediction uncertainty} \label{sec:approximationmethods}

In this section, we introduce the various approximation methods for the prediction uncertainty of single-output nonlinear regression models 
\begin{align} \label{eq:single-output-model}
f: \featurespace \times \parameterspace \to \Realspace
\end{align}
considered in this paper. In particular, we introduce our cubature-based approximation methods. All our approximation methods can be extended to the case of multi-output models in a straightforward manner, but for the sake of notational simplicity, we confine ourselves to single-output models~\eqref{eq:single-output-model} throughout this paper. As pointed out in the introduction, our interest in such approximations mainly comes from 
optimal experimental design \cite{Fedorov2014, Pronzato2013} where one tries to find experimental designs for which the maximal prediction uncertainty of the model trained on those designs becomes minimal. In essence, this prediction uncertainty comes from the fact that the experimental observations $\observations = (\target_1, \dots, \target_{\designsize})$ are subjected to random measurement errors. 
As is commonly done in both linear and nonlinear regression~\cite{Seber2003, Fedorov2014, Pronzato2013}, we assume (i) that these measurement errors are normally distributed 
and (ii) that the actual experimental observations 
$\observations$ are predicted -- up to the aforementioned measurement errors -- by the model $x \mapsto f(x,\true{\parameter})$ with some true parameter $\true{\parameter}$. Specifically, we assume that there exists a parameter value $\true{\parameter} \in \parameterspace$ such that for every $n \in \N$ and every set of input values $\feature_1, \dots, \feature_{\designsize} \in \featurespace$, the corresponding measured target values $\target_1, \dots, \target_{\designsize} \in \Realspace$ are given by the predictions $\model(\feature_1,\true{\parameter}), \dots, \model(\feature_{\designsize},\true{\parameter})$ of the model $f(\cdot,\true{\parameter})$ up to independent normally distributed measurement errors $\epsilon_1, \dots, \epsilon_{\designsize}$, that is, 
\begin{align} \label{eq:representability}
y_i = f(x_i,\theta^*) + \epsilon_i \qquad (i \in \{1,\dots,n\})
\end{align}
where $\epsilon_1, \dots, \epsilon_n$ are realizations of independent and normally distributed measurement errors $\varepsilon, \dots, \varepsilon_{\designsize}$ having mean $0$ and a known variance $\sigma^2$. A bit more concisely, this can be expressed by saying that the actually observed experimental results $\observations = (\target_1, \dots, \target_{\designsize})$ at the design $\design = (\feature_1, \dots, \feature_{\designsize})$ are realizations of the $\designsize$-dimensional normally distributed random variable
\begin{align} \label{eq:random-observation-model}
\observations\inb{\design} := \vectorized{\model}(\design, \true{\parameter}) + \vectorized{\varepsilon}\inb{\design}
\end{align}
with mean $\true{\observations} := \vectorized{\model}(\design, \true{\parameter})$ and covariance $\vectorized{\Sigma} := \diag(\sigma^2, \dots, \sigma^2) \in \Realspace^{\designsize \times \designsize}$. In this definition, 
\begin{align} \label{eq:predictions-of-true-model}
\true{\observations}
:= 
\vectorized{\model}(\design, \true{\parameter}) 
:=
\begin{pmatrix}
\model\inb{x_1, \true{\parameter}} \\
\vdots \\
\model\inb{x_{\designsize}, \true{\parameter}}
\end{pmatrix}
\in \Realspace^{\designsize}
\end{align}
represents the predictions of the true model at $\design$ and $\vectorized{\varepsilon}\inb{\design}$ is an $\designsize$-dimensional normally distributed random variable with mean $0$ and covariance $\vectorized{\Sigma}$, which represents the stacked measurement errors $\varepsilon_1, \dots, \varepsilon_{\designsize}$ at the individual design points $\feature_1, \dots, \feature_{\designsize}$. See~\cite{Bubel2024} for more background on this standard assumption from regression analysis. 
As an immediate consequence of this assumption, one obtains an integral representation of the prediction uncertainty
\begin{align} \label{eq:prediction-uncertainty}
\variance_{\design}\inb{\feature} 
:=
\varianceoperator{\model\inb{\feature, \estimator{\parameter}_{\model}\inb{\design, \vectorized{\target}\inb{\design}}}}  
\end{align}
of the model that has been trained based on the observation data $\observations\inb{\design}$ at the design $\design$. Indeed, 
\begin{align} \label{eq:prediction-uncertainty-integral-representation}
    \variance_{\design}\inb{\feature} = \variance_{\design, \true{\observations}}\inb{\feature} &= \varianceoperator{\model\inb{\feature, \estimator{\parameter}_{\model}\inb{\design, \vectorized{\target}\inb{\design}}}} \notag \\
    &= \int_{\Realspace^{\designsize}} \inb{\model\inb{\feature, \estimator{\parameter}_{\model}{\inb{\design, \vectorized{\target}}}} - \mean_{\design, \true{\observations}}\inb{\feature}}^2 \cdot \probabilityof{\vectorized{\target} - \true{\observations}} \d\vectorized{\target} \notag \\
    &= \int_{\Realspace^{\designsize}} \inb{\model\inb{\feature, \estimator{\parameter}_{\model}\inb{\design, \true{\observations} + \vectorized{\perturbation}}} - \mean_{\design, \true{\observations}}\inb{\feature}}^2 \cdot \probabilityof{\vectorized{\perturbation}} \d\vectorized{\perturbation},
\end{align}
where $\true{\observations} := \vectorized{\model}\inb{\design, \true{\parameter}}$ are the predictions of the true model at $\design$ as defined in~\eqref{eq:predictions-of-true-model}, $p$ is the probability density of the $\designsize$-dimensional normal distribution with mean $0$ and covariance matrix $\vectorized{\Sigma} = \diag\inb{\std^2, \dots, \std^2}$, and $\mean_{\design, \true{\observations}}\inb{\feature}$ denotes the expected prediction of the model that has been trained based on the observation data $\observations\inb{\design}$ at the design $\design$. In formulas,
\begin{align} \label{eq:expected-prediction}
    \mean_{\design}\inb{\feature} := \mean_{\design, \true{\observations}}\inb{\feature} &= \meanoperator{\model\inb{\feature, \estimator{\parameter}_{\model}\inb{\design, \vectorized{\target}\inb{\design}}}} = \int_{\Realspace^{\designsize}} \model\inb{\feature, \estimator{\parameter}_{\model}\inb{\design, \vectorized{\target}}} \cdot \probabilityof{\vectorized{\target} - \true{\observations}} \d\vectorized{\target} \notag \\
    &= \int_{\Realspace^{\designsize}} \model\inb{\feature, \estimator{\parameter}_{\model}\inb{\design, \true{\observations} + \vectorized{\perturbation}}} \cdot \probabilityof{\vectorized{\perturbation}} \d\vectorized{\perturbation}.
\end{align}
It is important to notice that the integral representation~\eqref{eq:prediction-uncertainty-integral-representation} can almost never be utilized directly for the computation of the prediction uncertainty of nonlinear models. In essence, this because (i) one almost never has a closed-form expression for the \leastsquares estimators $\estimator{\parameter}_{\model}(\design, \true{\observations}+\vectorized{\perturbation})$ of nonlinear models 
and because (ii) the true parameter value $\true{\parameter}$ and hence $\true{\observations} := \vectorized{\model}\inb{\design, \true{\parameter}}$ are unknown. It is therefore important to come up with suitable approximation methods for an approximate computation of the prediction uncertainty of nonlinear models. 
In the remainder of this section, we discuss such approximation methods. We begin with two well-known methods from the literature (\secn{sec:montecarloapproximation} and~\ref{sec:linearapproximation}) and then move on to introduce our cubature-based approximation methods (\secn{sec:cubatureapproximation}).

\subsection{Approximation using Monte Carlo sampling} \label{sec:montecarloapproximation}
Typically, no closed-form expression of the \leastsquares estimator, mean-, and variance of the prediction distribution is avilable for nonlinear regression models. Thus, the best approximation may be obtained using Monte Carlo sampling. This means, that observations from an experimental design are observed $\numberof{\montecarlo}$-times, where $\numberof{\montecarlo}$ is the number of Monte Carlo samples. For each of those $\numberof{\montecarlo}$ datasets 
\begin{align*}
(\design, \observations^{(1)}), \dots, (\design, \observations^{(\numberof{\montecarlo})})
\end{align*}
consisting of the same experimental design but individual noisy observations, a parameter estimator is computed. This gives the parameter estimator samples that are used to determine how "elliptic" the parameter estimator distribution is (elliptic sample distributions mean that the parameter nonlinearity in the model is weak).
The prediction uncertainty is then estimated by evaluating the model for each parameter estimator at a prediction location $\feature$ and by then using the empirical mean and variance as approximations to the true mean and variance of the prediction distribution. From \eqref{eq:expected-prediction} and \eqref{eq:prediction-uncertainty-integral-representation}, we obtain the Monte Carlo approximations as the mean
\begin{equation} \label{eq:montecarlomean}
    \mean_{\design}^{\montecarlo}\inb{\feature} := \frac{1}{\numberof{\montecarlo}} \sum_{i=1}^{\numberof{\montecarlo}} \model\inb{\feature, \estimator{\parameter}_{\model}\inb{\design, \observations^{(i)}}}
\end{equation}
and as the variance
\begin{equation} \label{eq:montecarlovariance}
    \variance_{\design}^{\montecarlo}\inb{\feature} := \frac{1}{\numberof{\montecarlo}} \sum_{i=1}^{\numberof{\montecarlo}} \inb{\model\inb{\feature, \estimator{\parameter}_{\model}\inb{\design, \observations^{(i)}}} - \mean_{\design}^{\montecarlo}\inb{\feature}}^2
\end{equation}
This allows for the computation of the mean \eqref{eq:montecarlomean} and variance \eqref{eq:montecarlovariance} of model prediction at an arbitrary precision, as, for the asymptotic case of $\numberof{\montecarlo} \to \infty$, the empiric mean and variance become the true mean and variance. In practice, however, we choose a finite number of Monte Carlo samples $\numberof{\montecarlo}$, which means that the approximation is only an approximation to the true mean and variance. This is why in \secn{sec:benchmarking}, we visualize the evolution of the Monte Carlo approximator to the prediction uncertainty over the number of samples. Note that the empirical mean of parameter estimators may introduce a bias and even for the asymptotic case, does not coincide with the true parameter value, as can be seen in \fig{fig:exptoymodel-parametersamples} (find a more detailed discussion on that in \secn{sec:benchmarking} for the exponential growth model).
Since the convergence of the Monte Carlo approximator is slow, we use sobol samples, transformed to a normal distribution, to achieve faster convergence (numerically) using pseudo-random number generators.

While being most accurate, a Monte Carlo-based prediction uncertainty estimator may only be feasible for computer experiments with a sufficiently small computational cost, it is typically infeasible in practice, where computer experiments are expensive or real-lab experiments are conducted. Also, in practice it is not always possible to perform repeatable experiments, as setting the experimental design may also come with some uncertainty. Therefore, there may be a mismatch between planned and actually performed experiments, which in the context of parameter estimation is referred to as error in variables, see e.g. Seber and Wild \cite{Seber2003}.

\subsection{Approximation using linearization} \label{sec:linearapproximation}
A very popular -- and in fact the standard -- method to approximate the prediction uncertainty of a nonlinear model is linearization. In this approach, one approximates the nonlinear model by a suitable linearized model and then exploits the fact that the prediction uncertainty for linear models can be expressed in closed form. Specifically, one linearizes the model $f$ around some reference parameter estimate $\meanof{\parameter}$ to obtain the linearized model $\model^{\lin}_{\meanof{\parameter}}$ defined by 
\begin{equation} \label{eq:linearizedmodel}
    \model^{\lin}_{\meanof{\parameter}}(\feature, \parameter) := \model\inb{\feature, \meanof{\parameter}} + \Jacobian_{\model}\inb{\feature, \meanof{\parameter}} \cdot \inb{\parameter - \meanof{\parameter}} \quad \inb{\inb{\feature, \parameter} \in \featurespace \times \Realspace^{\dimensionof{\parameter}}}.
\end{equation}
In the above equation, $\Jacobian_{\model}\inb{\feature, \meanof{\parameter}}$ is the Jacobian matrix of $\parameter \mapsto f(x,\parameter)$ at $\meanof{\parameter}$, that is, 
\begin{align*}
\Jacobian_{\model}\inb{\feature, \meanof{\parameter}} := \left[\frac{\partial \model\inb{\feature, \meanof{\parameter}}}{\partial \parameter_j}\right]_{j=1, 2, \dots, \dimensionof{\parameter}} \in \Realspace^{1 \times \dimensionof{\parameter}}
\end{align*}
Clearly, the linearized model $x \mapsto \model^{\lin}_{\meanof{\parameter}}(x,\parameter)$ is a good approximation to the true model $x \mapsto f(x,\true{\parameter})$ provided that $\parameter$ is sufficiently close to the parameter estimate $\meanof{\parameter}$ and that this parameter estimate, in turn, is sufficiently close to the true parameter value $\meanof{\parameter} \approx \true{\parameter}$. In general, this is only the case if one collects a large number of observations or if the observation noise is low.

It is well-known that the \leastsquares estimator of the linearized model can be expressed in closed form as
\begin{align} \label{eq:ls-estimate-linearized-model}
    \estimator{\parameter}_{\model^{\lin}_{\meanof{\parameter}}}\inb{\design, \observations} 
    = \meanof{\parameter} + \IM_{\model}\inb{\design, \meanof{\parameter}}^{-1} \cdot \vectorized{\Jacobian}_{\model}\inb{\design, \meanof{\parameter}}^T \cdot \vectorized{\Sigma}^{-1} \cdot \inb{\observations - \vectorized{\model}\inb{\design, \meanof{\parameter}}}
\end{align}
for arbitrary designs $\design$, observations $\observations$, and reference parameter estimates $\meanof{\parameter}$. See~\cite{Seber2003, Nocedal2006, Fedorov2014, Pronzato2013}, for instance. In the above equation, 
\begin{align} \label{eq:information-matrix}
\IM_{\model}\inb{\design, \meanof{\parameter}} = \vectorized{\Jacobian}_{\model}\inb{\design, \meanof{\parameter}}^T \cdot \vectorized{\covariance}^{-1} \cdot \vectorized{\Jacobian}_{\model}\inb{\design, \meanof{\parameter}}
\end{align}
is the so-called information matrix and
\begin{align*}
\vectorized{\Jacobian}_{\model}\inb{\design, \meanof{\parameter}} 
:= 
\begin{pmatrix}
\Jacobian_{\model}\inb{x_1, \meanof{\parameter}} \\
\vdots \\
\Jacobian_{\model}\inb{x_{\designsize}, \meanof{\parameter}}
\end{pmatrix}
\in \Realspace^{\designsize \times \dimensionof{\parameter}}
\quad \text{and} \quad
\vectorized{\model}(\design, \meanof{\parameter}) 
:=
\begin{pmatrix}
\model\inb{x_1, \meanof{\parameter}} \\
\vdots \\
\model\inb{x_{\designsize}, \meanof{\parameter}}
\end{pmatrix}
\in \Realspace^{\designsize}
\end{align*}
whereas $\vectorized{\Sigma} := \operatorname{\diag}(\sigma^2, \dots, \sigma^2) \in \Realspace^{\designsize \times \designsize}$. 
It is also well-known \cite{Seber2003, Nocedal2006, Fedorov2014, Pronzato2013} that the \leastsquares estimator of the linearized model, for suitable choices of the reference parameter, is a good approximation to the \leastsquares estimator of the nonlinear model. In short,
\begin{align} \label{eq:linearmodelparameterestimator_linearized}
    \estimator{\theta}_{\model}\inb{\design, \observations} \approx \estimator{\parameter}_{\model^{\lin}_{\meanof{\parameter}}}\inb{\design, \observations} 
\end{align}
for suitably chosen reference parameters $\meanof{\parameter}$. Indeed, if one iterates this approximation \eqref{eq:linearmodelparameterestimator_linearized} starting from an initial $\meanof{\parameter}^0$ and if in each iteration one chooses an appropriate step size, then one arrives at the well-known Gauß-Newton method. See \cite{Nocedal2006} for more information on the numerical procedure of \leastsquares estimation.

As the experimental observations $\observations$ for any given design $\design$ are subjected to random measurement errors, the corresponding \leastsquares parameter estimate $\estimator{\parameter}_{\model}\inb{\design, \observations}$ is a random variabble as well. A standard measure to quantify its uncertainty is the covariance matrix
\begin{equation} \label{eq:covariance-of-least-squares-estimator}
    \covarianceoperator{\estimator{\parameter}_{\model}\inb{\design, \observations\inb{\design}}} \in \Realspace^{\dimensionof{\parameter} \times \dimensionof{\parameter}}_{\pd}
\end{equation}
of the \leastsquares estimate for $\design$ and the random observations $\observations\inb{\design}$ as defined in~\eqref{eq:random-observation-model}. In view of \eqref{eq:linearmodelparameterestimator_linearized}, it is natural to approximate the covariance matrix of the least-squares estimator~\eqref{eq:covariance-of-least-squares-estimator} of the nonlinear model by the covariance matrix corresponding to the linearized model which by~\eqref{eq:ls-estimate-linearized-model}, in turn, can be represented in closed form as the inverse information matrix. In short, 
\begin{equation} \label{eq:linearizedmodelparametercovariance}
    \covarianceoperator{\estimator{\parameter}_{\model}\inb{\design, \observations\inb{\design}}} 
    \approx \covarianceoperator{\estimator{\parameter}_{\model^{\lin}_{\meanof{\parameter}}}\inb{\design, \observations\inb{\design}}} 
    = \IM_{\model}\inb{\design, \meanof{\parameter}}^{-1}
\end{equation}
for suitably chosen $\meanof{\parameter} \in \parameterspace$. In contrast to linear models, the information matrix $\IM_{\model}\inb{\feature, \meanof{\parameter}}$ for nonlinear models may strongly depend on $\meanof{\parameter}$. And therefore, choosing a suitable value for $\meanof{\parameter}$ is of critical importance, in general. Conventionally, one chooses the reference parameter $\meanof{\parameter} := \estimator{\parameter}_{\model}\inb{\design, \observations}$ to be the \leastsquares estimator for the available observations $\observations$ of $\model$ at $\design$. 

With such observations at hand, one obtains a simple and very popular approximation to the prediction uncertainty $\variance_{\design}$, namely the linearization approximation 
\begin{align} \label{eq:prediction-uncertainty-linearization}
\variance^{\lin}_{\design,\observations} \approx \variance_{\design},
\end{align}
by simply replacing the nonlinear model $f$ in~\eqref{eq:prediction-uncertainty} by its linearization $\model^{\lin}_{\meanof{\parameter}}$ around $\meanof{\parameter}$. Specifically, 
\begin{align} \label{eq:linearizationvariance}
\variance^{\lin}_{\design,\observations}(\feature) 
:= 
\varianceoperator{\model^{\lin}_{\meanof{\parameter}} \inb{\design, \estimator{\parameter}_{\model^{\lin}_{\meanof{\parameter}}}\inb{\design, \observations\inb{\design}}}}
\end{align}
with $\meanof{\parameter} := \estimator{\parameter}_{\model}\inb{\design, \observations}$. In view of~\eqref{eq:linearizedmodel} and~\eqref{eq:linearizedmodelparametercovariance}, one then obtains the following well-known closed-form representation of~\eqref{eq:linearizationvariance} in terms of Jacobian and information matrices:
\begin{align}
\variance^{\lin}_{\design,\observations}(\feature) 
= 
\Jacobian_{\model}\inb{\feature, \meanof{\parameter}} \IM_{\model}\inb{\feature, \meanof{\parameter}}^{-1} \Jacobian_{\model}\inb{\feature, \meanof{\parameter}}^T,
\end{align}
where $\meanof{\parameter} := \estimator{\parameter}_{\model}\inb{\design, \observations}$. Clearly, the twofold use of linearization in \eqref{eq:linearizationvariance} introduces an error to the approximation~\eqref{eq:prediction-uncertainty-linearization}, which for a highly nonlinear model $\model$ can be very large -- even in the asymptotic case where $\meanof{\parameter} \approx \true{\parameter}$. Additionally, $\meanof{\parameter} \approx \true{\parameter}$ is not a realistic assumption in practice, especially in the context of (model-based) optimal experimental design, where one typically is in the early stages of model calibration. This means that one does not have many data yet, and one further seeks to generate minimum-redundancy designs which yield a maximum reduction in the maximum prediction error of the considered nonlinear model.

\subsection{Approximations using cubature formulas} \label{sec:cubatureapproximation}
After having recalled the most common approximation methods for prediction uncertainties of nonlinear regression models, we now introduce the novel approximation methods of this paper, which are based on appropriate cubature methods. 
Cubature methods are approaches to 
approximate multivariate integrals
\begin{align} \label{eq:cubature-integral}
\int_{\Realspace^{\designsize}} h(\vectorized{\perturbation}) p(\vectorized{\perturbation}) \d \vectorized{\perturbation}
\end{align}
by means of finite sums of the form
\begin{align} \label{eq:cubature-formula}
\sum_{i=1}^N w_i \cdot h(\vectorized{\perturbation}^{(i)})
\end{align}
with a small or moderate number $N$ of cubature points $\vectorized{\perturbation}^{(1)}, \dots, \vectorized{\perturbation}^{(N)}$ and weights $w_1, \dots w_N$. (In the special case of univariate integrals $\designsize = 1$, cubature methods are usually referred to as quadrature methods, but this special case plays no role in our context where $\designsize$ is the experimental design size.) We refer to~\cite{Engels1980, Davis1984, Hinrichs2007, Orive2021} for comprehensive introductions to the field of cubature methods.
Since the prediction uncertainty~\eqref{eq:prediction-uncertainty} that we are interested in here 
actually is an integral of the form~\eqref{eq:cubature-integral}, it is natural to attempt 
its approximate computation with the help of cubature formulas. In the following, we therefore propose various cubature-based approximation methods for the prediction uncertainty. 

We begin with a basic approximation method based on the simple but relatively well-known sigma points~\cite{Julier1996, Julier2004}. We then improve this sigma point approximation by bringing to bear 
more advanced cubature methods for symmetric probability densities $p$, namely the ones by McNamee and Stenger~\cite{McNamee1967} and, respectively, by Lu and Darmofal~\cite{Lu2004}.

\subsubsection{Cubature formula based on sigma points} \label{sec:sigmapointapproximation}
As a first approximation method for prediction uncertainties, we propose to use sigma points as cubature points. Sigma points have originally been developed by Julier and Uhlmann~\cite{Julier1996, Julier2004} 
to approximate the variance -- or, more generally, the covariance -- of nonlinear transformations of normally distributed random variables. We apply them to the predictions $\model(\feature, \estimator{\parameter}_{\model}(\design, \observations\inb{\design}))$ of the model trained on $\design$ and $\observations\inb{\design}$, which are precisely such nonlinear transformations of the normally distributed random observations $\observations\inb{\design}$ (with the corresponding nonlinear transformations being given by the maps $\Realspace^{\designsize} \ni \observations \mapsto \model(\feature, \estimator{\parameter}_{\model}(\design, \observations))$, of course). 
In essence, the idea behind sigma points is to approximate the variance, in our case the prediction uncertainty
\begin{align}
	\variance_{\design}\inb{\feature} 
	=
	\varianceoperator{\model\inb{\feature, \estimator{\parameter}_{\model}\inb{\design, \vectorized{\target}\inb{\design}}}}, 
\end{align}
by a suitably weighted empirical variance at suitably chosen sample points $\vectorized{\target}_{\pm}^{(1)}, \dots, \vectorized{\target}_{\pm}^{(\designsize)}$, the so-called sigma points. Specifically, these points are obtained by systematically perturbing all $\designsize$ components of the prediction $\true{\observations} = \vectorized{\model}\inb{\design, \true{\parameter}}$ of the true model by some increment $\delta$ up- and downwards. In formulas, 
\begin{gather} 
	\vectorized{\target}_{\pm}^{(i)} := \true{\observations} \pm \delta \cdot \unitvector_i
	\qquad (i \in \{1,\dots,\designsize\}) 
\label{eq:sigma-points} \\
	\delta := \sqrt{\designsize + \kappa} \cdot \sigma
\label{eq:increment-sigma-points}
\end{gather} 
where $\kappa \in (-n,\infty)$ is a hyperparameter. See~\cite{Julier1996, Julier2004}.
In practice, however, the predictions $\true{\observations} = \vectorized{\model}\inb{\design, \true{\parameter}}$ of the true model cannot be assumed to be known, just because the true parameter is not known. And therefore, we propose to use 
\begin{align} \label{eq:predictions-of-estimated-model}
\vectorized{\meanof{\target}} := \vectorized{\model}(\design, \meanof{\parameter}) := \vectorized{\model}(\design, \estimator{\parameter}_{\model}(\design, \observations))
\end{align}
instead, that is, the predictions of the estimated model trained on the available experimental observations $\observations$ at $\design$. With this modification, we arrive at the following sigma point approximation $\variance_{\design, \observations}^{\sigmapoint}\inb{\feature}$ for the prediction uncertainty $\variance_{\design}\inb{\feature}$:
\begin{align} \label{eq:sigma-point-prediction-uncertainty}
	\variance_{\design, \observations}^{\sigmapoint}\inb{\feature} := w_0 \cdot h_{\design,\observations}^{\sigmapoint}(x,0) + w_1 \cdot \sum_{i=1}^{\designsize} \inb{h_{\design,\observations}^{\sigmapoint}(x, \delta \cdot \unitvector_i) + h_{\design,\observations}^{\sigmapoint}(x, -\delta \cdot \unitvector_i)}, 
\end{align}
where the weights $w_0$ and $w_1$ are defined to be 
\begin{align} \label{eq:weights-sigma-points}
w_0 := \frac{\kappa}{\designsize + \kappa}
\qquad \text{and} \qquad
w_1 := \frac{1}{2(\designsize + \kappa)}.
\end{align}
Additionally, $h_{\design,\observations}^{\sigmapoint}(x,\vectorized{\perturbation})$ is a shorthand for the squared deviation
\begin{align} \label{eq:squared-error-sigma-points}
h_{\design,\observations}^{\sigmapoint}(x,\vectorized{\perturbation}) := \inb{\model\inb{\feature, \estimator{\parameter}_{\model}\inb{\design, \vectorized{\meanof{\target}} + \vectorized{\perturbation}}} - \mean_{\design, \observations}^{\sigmapoint}\inb{\feature}}^2
\end{align}
and $\mean_{\design, \observations}^{\sigmapoint}\inb{\feature}$, in turn, is the sigma-point approximation for the expected prediction~\eqref{eq:expected-prediction}, that is, 
\begin{align} \label{eq:sigma-point-expected-prediction}
	\mean_{\design, \observations}^{\sigmapoint}\inb{\feature} := w_0 \cdot g_{\design,\observations}(x,0) + w_1 \cdot \sum_{i=1}^{\designsize} \inb{g_{\design,\observations}(x, \delta \cdot \unitvector_i) + g_{\design,\observations}(x, -\delta \cdot \unitvector_i)}
\end{align}
where $g_{\design,\observations}(x,\vectorized{\perturbation})$ is the prediction of the model trained on the perturbed observations $\vectorized{\meanof{\target}} + \vectorized{\perturbation}$:
\begin{align} \label{eq:model-trained-on-perturbed-observations} 
g_{\design,\observations}(x,\vectorized{\perturbation}) := \model\inb{\feature, \estimator{\parameter}_{\model}\inb{\design, \vectorized{\meanof{\target}} + \vectorized{\perturbation}}}.
\end{align}

In view of~\eqref{eq:prediction-uncertainty-integral-representation} and~\eqref{eq:sigma-point-prediction-uncertainty}, it is clear that the sigma-point approximation $\variance_{\design, \observations}^{\sigmapoint}$ to the prediction uncertainty is a cubature formula in the classical sense~\eqref{eq:cubature-formula} with 
\begin{align}
N_{\sigmapoint} = 2\designsize + 1
\end{align}
cubature points, namely $0$ and $\pm \delta \unitvector_1, \dots, \pm \delta \unitvector_{\designsize}$. Additionally, it is clear from the definitions~\eqref{eq:squared-error-sigma-points} and~\eqref{eq:model-trained-on-perturbed-observations} that computing $\variance_{\design, \observations}^{\sigmapoint}(\feature)$ requires the computation of $2 \designsize + 1$ least-squares estimates, namely one for each sigma point.  

An obvious shortcoming of the sigma-point approximation $\variance_{\design, \observations}^{\sigmapoint}$ to the prediction uncertainty is that it depends on the hyperparameter $\kappa$ through the increment~\eqref{eq:increment-sigma-points} and the weights~\eqref{eq:weights-sigma-points}. In spite of some heuristic rules of thumb~\cite{Julier1996, Julier2004}, it is generally unclear how to choose this hyperparameter. In our experience, this value has to be chosen differently from application to application to prevent poor approximtions, and this requires quite some tuning effort.  
Another shortcoming of the sigma-point approximation is that it does not come with (strong) theoretical exactness guarantees. Indeed, we will show that the sigma-point approximation $\variance_{\design, \observations}^{\sigmapoint}$ is not exact even for quadratic models $\model$ with completely noisefree observations $\observations = \true{\observations}$ (Corollary~\ref{cor:comparison-of-prediction-uncertainty-measures-toy-example-multivariate}). 
We therefore propose more advanced approximations now, which come without the aforementioned drawbacks. 

\subsubsection{Cubature formula of McNamee and Stenger} \label{sec:mcnameestengerapproximation}
As a first alternative to the sigma-point approximation~\eqref{eq:sigma-point-prediction-uncertainty} for the prediction uncertainty, we propose an approximation based on the cubature formulas by McNamee and Stenger~\cite{McNamee1967}. 
Specifically, we apply the fifth-degree cubature formula from \cite{McNamee1967} (Section~4.1) to the integral representations~\eqref{eq:prediction-uncertainty-integral-representation} and~\eqref{eq:expected-prediction} for the prediction uncertainty and the expected prediction and then replace the unknown value $\true{\observations} := \vectorized{\model}\inb{\design, \true{\parameter}}$ featuring in \eqref{eq:expected-prediction} and \eqref{eq:prediction-uncertainty-integral-representation} by the predictions $\vectorized{\meanof{\target}}$ of the estimated model as defined in~\eqref{eq:predictions-of-estimated-model}. In this manner, we arrive at the approximation formula
\begin{align} \label{eq:prediction-uncertainty-mcnamee-stenger}
    &\variance_{\design, \observations}^{\mcnameestenger}(\feature) := w_0 \cdot h_{\design,\observations}^{\mcnameestenger}(x,0) + w_1 \cdot \sum_{i=1}^{\designsize} \inb{h_{\design,\observations}^{\mcnameestenger}(x, \delta \cdot \unitvector_i) + h_{\design,\observations}^{\mcnameestenger}(x, -\delta \cdot \unitvector_i)} \notag \\
    & + w_2 \cdot \sum_{i=1}^{\designsize} \sum_{j=1}^{i-1} \big( h_{\design,\observations}^{\mcnameestenger}(x, \delta \cdot (\unitvector_i + \unitvector_j)) + h_{\design,\observations}^{\mcnameestenger}(x, -\delta \cdot (\unitvector_i + \unitvector_j)) \notag \\
    & \qquad \qquad \qquad + h_{\design,\observations}^{\mcnameestenger}(x, \delta \cdot (\unitvector_i - \unitvector_j)) + h_{\design,\observations}^{\mcnameestenger}(x, -\delta \cdot (\unitvector_i - \unitvector_j)) \big)
\end{align}
for the prediction uncertainty $\variance_{\design, \observations}(\feature)$, where $e_i$ denotes the $i$th canonical unit vector in $\Realspace^{\designsize}$ and the weights $w_0, w_1, w_2$ and the increment $\delta$ are defined as
\begin{gather}
    w_0 := I_0 - \designsize \cdot \inb{I_2 / I_4}^2 \cdot \inb{I_4 - \frac{\designsize - 1}{2}} \cdot I_{2,2} \\
    w_1 := \frac{\inb{I_2 / I_4}^2}{2} \cdot \inb{I_4 - \inb{\designsize - 1} \cdot I_{2,2}} \qquad \text{and} \qquad
    w_2 := \frac{\inb{I_2 / I_4}^2}{4} \cdot I_{2,2} \\
    \delta  := \inb{I_4 / I_2}^{1/2} = \sqrt{3} \cdot \std
\end{gather}
with $I_0 = 1$, $I_2 = \std^2$, $I_4 = 3 \cdot \std^4$, and $I_{2,2} = \std^4$. Additionally, $h_{\design,\observations}^{\mcnameestenger}(x,\vectorized{\perturbation})$ is a shorthand for the squared deviation
\begin{align} \label{eq:squared-error-mcnamee-stenger}
h_{\design,\observations}^{\mcnameestenger}(x,\vectorized{\perturbation}) := \inb{\model\inb{\feature, \estimator{\parameter}_{\model}\inb{\design, \vectorized{\meanof{\target}} + \vectorized{\perturbation}}} - \mean_{\design, \observations}^{\mcnameestenger}\inb{\feature}}^2
\end{align}
and $\mean_{\design, \observations}^{\mcnameestenger}\inb{\feature}$, in turn, is the approximation for the expected prediction~\eqref{eq:expected-prediction} formed completely analogously to~\eqref{eq:prediction-uncertainty-mcnamee-stenger}, that is, by replacing $h_{\design,\observations}^{\mcnameestenger}$ in~\eqref{eq:prediction-uncertainty-mcnamee-stenger} by the perturbed predictions $g_{\design,\observations}$ from~\eqref{eq:model-trained-on-perturbed-observations}. 
Counting the summands in~\eqref{eq:prediction-uncertainty-mcnamee-stenger}, we see that the McNamee-Stenger approximation requires considerably more cubature points -- and hence least-squares estimators -- than the sigma-point approximation~\eqref{eq:sigma-point-prediction-uncertainty}, namely
\begin{align} \label{eq:N_MS}
N_{\mcnameestenger} = 1 + 2\designsize + 4 \designsize (\designsize -1)/2 = 2 \designsize^2 + 1.
\end{align}
In return, however, the McNamee-Stenger approximation is guaranteed to be exact up to fifth degree in $h_{\design,\observations}^{\mcnameestenger}(x,\cdot)$ (Section~4 in~\cite{McNamee1967}), while the sigma-point approximation does not come with such guarantees.
Additionally, the McNamee-Stenger approximation does not contain artificial tuning parameters like the hyperparameter $\kappa$ from~\eqref{eq:increment-sigma-points} and~\eqref{eq:weights-sigma-points}. 

\subsubsection{Cubature formula of Lu and Darmofal} \label{sec:ludarmofalapproximation}
As a second alternative to the sigma-point approximation~\eqref{eq:sigma-point-prediction-uncertainty} for the prediction uncertainty, we propose an approximation based on the cubature formulas by Lu and Darmofal~\cite{McNamee1967}, which improve upon~\cite{McNamee1967} in several ways.  
Specifically, we apply the fifth-degree cubature formula from \cite{Lu2004} (Section~4.1) to the integral representations~\eqref{eq:prediction-uncertainty-integral-representation} and~\eqref{eq:expected-prediction} for the prediction uncertainty and the expected prediction and then replace the unknown value $\true{\observations} := \vectorized{\model}\inb{\design, \true{\parameter}}$ featuring in \eqref{eq:expected-prediction} and \eqref{eq:prediction-uncertainty-integral-representation} by the predictions $\vectorized{\meanof{\target}}$ of the estimated model as defined in~\eqref{eq:predictions-of-estimated-model}. In this manner, we arrive at the approximation formula
\begin{align} \label{eq:prediction-uncertainty-lu-darmofal}
    &\variance_{\design, \observations}^{\ludarmofal}\inb{\feature} := w_0 \cdot h_{\design,\observations}^{\ludarmofal}(x,0) + w_1 \cdot \sum_{i=1}^{\designsize+1} \inb{h_{\design,\observations}^{\ludarmofal}(x,\delta \cdot a^{(i)}) + h_{\design,\observations}^{\ludarmofal}(x,-\delta \cdot a^{(i)})} \notag \\ 
    &\qquad + w_2 \cdot \sum_{i=1}^{\designsize+1} \sum_{j=1}^{i-1} \inb{h_{\design,\observations}^{\ludarmofal}(x,\delta \cdot b^{(i,j)}) + h_{\design,\observations}^{\ludarmofal}(x,-\delta \cdot b^{(i,j)})}
\end{align}
for the prediction uncertainty $\variance_{\design, \observations}(\feature)$, where the weights $w_0, w_1, w_2$ and the increment $\delta$ are defined as
\begin{gather} \label{eq:ludarmofalparameters}
    w_0 := \frac{2}{\designsize + 2} \qquad \text{and} \qquad
    w_1 := \frac{\designsize^2 \cdot \inb{7 - \designsize}}{2 \cdot \inb{\designsize + 1}^2 \cdot \inb{\designsize + 2}^2} \\
    w_2 := \frac{2 \cdot \inb{\designsize - 1}^2}{\inb{\designsize + 1}^2 \cdot \inb{\designsize + 2}^2} \qquad \text{and} \qquad
    \delta := \sqrt{n+2} \cdot \sigma
\end{gather}
and the perturbation directions $a^{(i)}, b^{(i,j)} \in \Realspace^{\designsize}$ defining the cubature points are given by
\begin{equation}
    a^{(i)}_k := \begin{cases}
        - \sqrt{\frac{\designsize + 1}{\designsize \cdot \inb{\designsize - k + 2} \cdot \inb{\designsize -k + 1}}}, \qquad k < i \\
        \sqrt{\frac{\inb{\designsize + 1} \cdot \inb{\designsize - i + 1}}{\designsize \cdot \inb{\designsize - i + 1}}}, \qquad k = i \\
        0, \qquad k > i
    \end{cases}
\end{equation}
for all $k \in \left\{1, 2, \dots, \designsize\right\}$ and, respectively, by
\begin{equation}
    b^{(i,j)} := \sqrt{\frac{\designsize}{2 \cdot \inb{\designsize - 1}}} \cdot \inb{a^{(i)} + a^{(j)}}
\end{equation}
for all $i < j$. Additionally, $h_{\design,\observations}^{\ludarmofal}(x,\vectorized{\perturbation})$ is a shorthand for the squared deviation
\begin{align} \label{eq:squared-error-lu-darmofal}
h_{\design,\observations}^{\ludarmofal}(x,\vectorized{\perturbation}) := \inb{\model\inb{\feature, \estimator{\parameter}_{\model}\inb{\design, \vectorized{\meanof{\target}} + \vectorized{\perturbation}}} - \mean_{\design, \observations}^{\ludarmofal}\inb{\feature}}^2
\end{align}
and $\mean_{\design, \observations}^{\ludarmofal}\inb{\feature}$, in turn, is the approximation for the expected prediction~\eqref{eq:expected-prediction} formed completely analogously to~\eqref{eq:prediction-uncertainty-lu-darmofal}, that is, by replacing $h_{\design,\observations}^{\ludarmofal}$ in~\eqref{eq:prediction-uncertainty-lu-darmofal} by the perturbed predictions $g_{\design,\observations}$ from~\eqref{eq:model-trained-on-perturbed-observations}.
Clearly, the Lu-Darmofal approximation still requires considerably more cubature points than the sigma-point approximation~\eqref{eq:sigma-point-prediction-uncertainty}, but only about half as many as the McNamee-Stenger approximation, namely
\begin{align} \label{eq:N_LD}
N_{\ludarmofal} = 1 + 2(\designsize + 1) + 2 \designsize (\designsize + 1)/2 = \designsize^2 + 3\designsize + 3.
\end{align}
Also, the Lu-Darmofal approximation is guaranteed to be exact up to fifth degree in $h_{\design,\observations}^{\ludarmofal}(x,\cdot)$ (Section 4 in~\cite{Lu2004}) and the number $N_{\ludarmofal}$ is close to the minimal possible number $\designsize^2 + \designsize + 1$ of cubature points with the aforementioned exactness property (Section 2 and Figure 1 in~\cite{Lu2004}).

\section{Validation and benchmarking} \label{sec:validationandbenchmarking}
In this section, we validate the proposed cubature-based approximation methods for the prediction uncertainty by means of increasingly complex models on different experimental designs. We begin with a generic quadratic model, for which we establish closed-form expressions for the prediction uncertainty and all the considered approximations (\secn{sec:validation}). We then validate our approximation methods on more general models and experimental designs, for which no closed-form representations can be given anymore (\secn{sec:benchmarking}). Specifically, we consider an exponential growth model and the NRTL model. 

\subsection{Validation on a generic quadratic model} \label{sec:validation}
In this section, we consider a generic quadratic model which is quadratic both in the inputs $\feature$ and in the model parameters $\parameter$. As we will see, the McNamee-Stenger and Lu-Darmofal approximations~\eqref{eq:prediction-uncertainty-mcnamee-stenger} and~\eqref{eq:prediction-uncertainty-lu-darmofal} are exact for this model on orthogonal factorial designs, whereas the linearization and sigma-point approximations are not (Corollary~\ref{cor:comparison-of-prediction-uncertainty-measures-toy-example-multivariate}).
Specifically, we consider the separable quadratic model $\model: \featurespace \times \parameterspace \to \Realspace$ defined by
\begin{align} \label{eq:multivariatetoymodel}
    \model\inb{\feature, \parameter} &:= \model_{\alpha, \beta}\inb{\feature, \parameter} \notag\\
    &:= \parameter_0 + \sum_{k=1}^{\dimensionof{\feature}} \alpha_k \cdot \parameter_k \cdot \feature_k + \sum_{k=1}^{\dimensionof{\feature}} \beta_k \cdot \frac{\parameter_k^2}{2} \cdot \feature_k^2 \quad \inb{\inb{\feature, \parameter} \in \featurespace \times \parameterspace}
\end{align}
with input space $\featurespace := \left[-1, 1\right]^{\dimensionof{\feature}}$ and parameter space $\parameterspace := \Realspace^{\dimensionof{\parameter}} := \Realspace^{\dimensionof{\feature}+1}$ and with arbitrary linear and quadratic coefficients 
\begin{align}
\alpha_k \in \Realspace \setminus \{0\} \qquad \text{and} \qquad \beta_k \in \Realspace
\end{align}
(hyperparameters). We will show that the prediction uncertainty of this model can be computed explicitly, provided that we consider the specific experimental design $\design := \inb{\feature_1, \feature_2, \dots, \feature_{\designsize}}$ whose design points $\feature_1, \dots, \feature_{\designsize}$ are corner points of the input space with mean $0$ and covariance matrix $n \cdot \unitmatrix_{\dimensionof{\feature} \times \dimensionof{\feature}}$. Spelled out, this means that we consider experimental designs $\design := \inb{\feature_1, \dots, \feature_{\designsize}}$ statisfying the following conditions:
\begin{equation} \label{eq:multivariatetoydesigncondition1}
    \feature_i \in \left\{-1, 1\right\}^{\dimensionof{\feature}} \quad \inb{i \in \left\{1, 2, \dots, \designsize\right\}} 
    \qquad \text{and} \qquad \sum_{i=1}^{\designsize} \feature_i  = 0
\end{equation}
\begin{equation} \label{eq:multivariatetoydesigncondition2}
    \sum_{i=1}^{\designsize} \feature_{i, k} \cdot \feature_{i, l} = \designsize \cdot \delta_{k, l} \quad \inb{k, l \in \left\{1, 2, \dots, \dimensionof{\feature}\right\}}.
\end{equation}
In particular, \eqref{eq:multivariatetoydesigncondition1} implies that the design size $\designsize$ is an even number. Also, \eqref{eq:multivariatetoydesigncondition2} says that the considered designs are orthogonal designs (see page 459 of \cite{Montgomery2009}, for instance). It is trivial to verify that conditions \eqref{eq:multivariatetoydesigncondition1} and \eqref{eq:multivariatetoydesigncondition2} are satisfied for the full factorial design $\design = \inb{\feature_1, \feature_2, \dots, \feature_{\designsize}}$ consisting of the $2^{\dimensionof{\feature}}$ corners of the input space $\featurespace = \left[-1, 1\right]^{\dimensionof{\feature}}$, that is, $\designsize = 2^{\dimensionof{\feature}}$ and
\begin{align*}
    \left\{\feature_i : i \in \left\{1, 2, \dots, 2^{\dimensionof{\feature}}\right\}\right\} = \left\{\inb{\inb{-1}^{k_1}, \dots, \inb{-1}^{k_{\dimensionof{\feature}}}}: k_1, \dots, k_{\dimensionof{\feature}} \in \left\{0, 1\right\}\right\}.
\end{align*}
And, of course, the conditions \eqref{eq:multivariatetoydesigncondition1} and \eqref{eq:multivariatetoydesigncondition2} are also satisfied for every replicated version of this full factorial design. Additionally, the conditions \eqref{eq:multivariatetoydesigncondition1} and \eqref{eq:multivariatetoydesigncondition2} are satisfied for appropriate fractions of the full factorial design above, that is, for appropriate fractional factorial designs. Indeed, in the case $\dimensionof{\feature} = 3$, a fractional factorial design satisfying \eqref{eq:multivariatetoydesigncondition1} and \eqref{eq:multivariatetoydesigncondition2} is given by $\design = \inb{\feature_1, \feature_2, \dots, \feature_{\designsize}}$ with
\begin{equation} \label{eq:multivariatetoydesign}
    \inb{\feature_1, \dots, \feature_4} 
    :=
        \begin{pmatrix}
            -1 & 1 & -1 & 1 \\
            -1 & -1 & 1 & 1 \\
            -1 & 1 & 1  & -1
        \end{pmatrix},
\end{equation}
for instance. In the case, $\dimensionof{\feature}=2$, however, only the full factorial design and its replicated versions satisfy \eqref{eq:multivariatetoydesigncondition1} and \eqref{eq:multivariatetoydesigncondition2}.

As can be easily shown, for every design $\design$ satisfying~\eqref{eq:multivariatetoydesigncondition1} and~\eqref{eq:multivariatetoydesigncondition2} and arbitrary observations $\observations = \inb{\target_1, \dots, \target_{\designsize}} \in \Realspace^{\designsize}$, there exists a unique \leastsquares estimator $\estimator{\parameter}_{\model}\inb{\design, \observations} = (\estimator{\parameter}_0, \estimator{\parameter}_1, \dots, \estimator{\parameter}_{\dimensionof{\feature}})$, and it is given by
\begin{equation}
\begin{aligned} \label{eq:multivariatetoymodelexactestimator}
    \estimator{\parameter}_0 &= \frac{1}{\designsize} \sum_{i=1}^{\designsize} \target_i - \sum_{k=1}^{\dimensionof{\feature}} \frac{\beta_k}{2 \cdot \alpha_k^2 \cdot \designsize^2} \cdot \inb{\sum_{i=1}^{\designsize} \feature_{i, k} \cdot \target_i}^2, \\ 
    \estimator{\parameter}_k &= \frac{1}{\alpha_k \cdot \designsize} \sum_{i=1}^{\designsize} \feature_{i, k} \cdot \target_i
    \qquad (k \in \left\{1, \dots, \dimensionof{\feature}\right\})
\end{aligned}
\end{equation}
(Lemma~\ref{lm:least-squares-estimator-toy-example-multivariate}). With the help of these explicit formulas, we can then establish a closed-form identity for the prediction uncertainty $\variance_{\design}\inb{\feature}$ of the quadratic model at any input point $\feature = (\feature_k)_{k=1,\dots,\dimensionof{\feature}} \in \featurespace$, namely
\begin{align} \label{eq:multivariateanalyticalmodelvariance}
    \variance_{\design}\inb{\feature} &= \frac{\std^2}{\designsize} \cdot \inb{1 + \sum_{k=1}^{\dimensionof{\feature}} \inb{\feature_k + \frac{\beta_k}{\alpha_k} \cdot \inb{\feature_k^2 - 1} \cdot \true{\parameter}_k}^2} \notag \\
    &\quad + \frac{\std^4}{2 \cdot \designsize^2} \cdot \sum_{k=1}^{\dimensionof{\feature}} \frac{\beta_k^2}{\alpha_k^4} \cdot \inb{\feature_k^2 - 1}^2
\end{align}
(Corollary~\ref{cor:exact-prediction-uncertainty-toy-example-multivariate}). We now use this explicit formula to validate the proposed cubature-based approximation methods~\eqref{eq:sigma-point-prediction-uncertainty}, \eqref{eq:prediction-uncertainty-mcnamee-stenger}, \eqref{eq:prediction-uncertainty-lu-darmofal} and to compare them to the linearization-based approximation~\eqref{eq:prediction-uncertainty-linearization}.  Specifically, we investigate how far the approximate prediction uncertainty $\variance^{\mathrm{method}}_{\design, \observations}\inb{\feature}$ for each of the considered approximation methods deviates from the exact prediction uncertainty $\variance_{\design}\inb{\feature}$ as given by~\eqref{eq:multivariateanalyticalmodelvariance}. In other words, we consider the approximation error
\begin{equation} \label{eq:validationloss}
   \Delta_{\design, \observations}^{\mathrm{method}}\inb{\feature} 
   := | \variance^{\mathrm{method}}_{\design, \observations}\inb{\feature} - \variance_{\design}\inb{\feature} |
\end{equation}
for each of the considered approximation methods linearization ($\lin$), sigma-point ($\sigmapoint$), McNamee-Stenger ($\mcnameestenger$), Lu-Darmofal ($\ludarmofal$), as a function of the input $\feature \in \featurespace$. 

\begin{figure}
    \centering
    \includegraphics[width=1.0\textwidth]{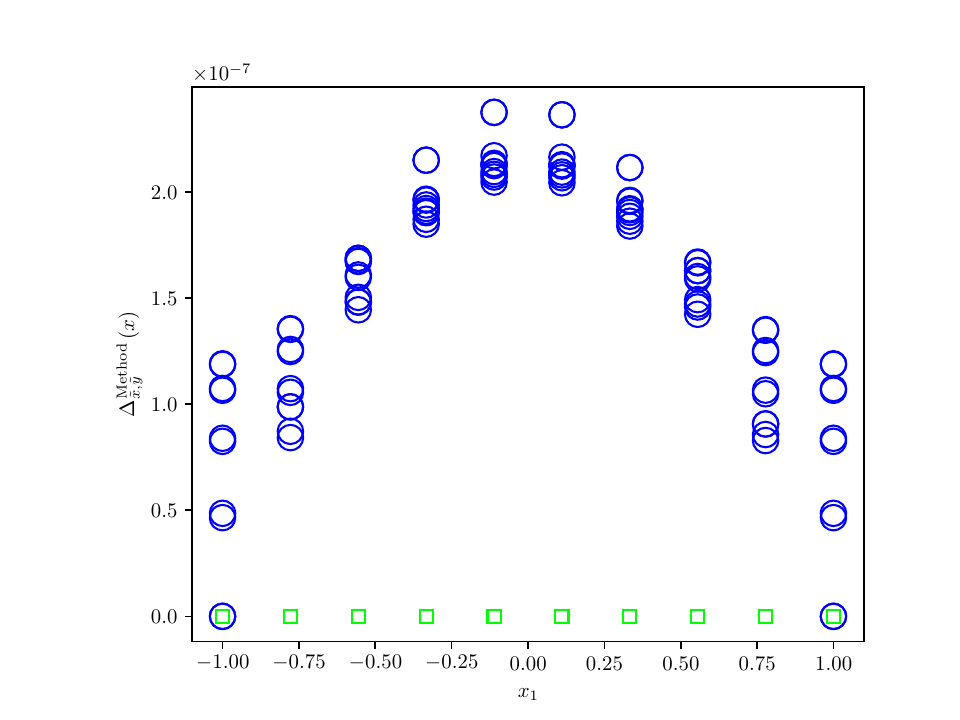}
    \caption{Approximation error \eqref{eq:validationloss} for the quadratic model~\eqref{eq:multivariatetoymodel} in $\dimensionof{\feature} = 2$ input dimensions with coefficients~\eqref{eq:model-coefficients-validation} and with the factorial design~\eqref{eq:factorial-design-for-validation}. The approximation error is plotted as a function $\feature_1 \mapsto \Delta_{\design, \observations}^{\mathrm{method}}\inb{\feature_1, \feature_2}$ of $\feature_1$ for several fixed values of $\feature_2$ and for $\mathrm{method} = \lin$ (blue circles) and $\mathrm{method} = \ludarmofal$ (lime squares).}
    \label{fig:multivariatetoymodel-validation}
\end{figure}

Specifically, we conduct a benchmark comparing the linearization-based approximation method with the cubature-based approximation method of Lu-Darmofal for the quadratic model \eqref{eq:multivariatetoymodel} in $\dimensionof{\feature} = 2$ input dimensions with the linear and quadratic coefficients chosen to be
\begin{align} \label{eq:model-coefficients-validation}
\alpha_k = 1 \qquad \text{and} \qquad \beta_k = 1
\end{align}
for $k \in \{1,2\}$. As the design, we choose the factorial design of $\designsize=8$ points given by
\begin{align} \label{eq:factorial-design-for-validation}
    \design 
    := \inb{\feature_1, \feature_2, \dots, \feature_8}
    := \begin{pmatrix}
        - 1 & -1 & 1  & 1 & -1 & -1 & 1  & 1 \\
        - 1 & 1  & -1 & 1 & -1 & 1  & -1 & 1  
    \end{pmatrix}.
\end{align}
It includes each vertex of the input space $\featurespace = \left[-1, 1\right]^2$ exactly twice and clearly satisfies the conditions~\eqref{eq:multivariatetoydesigncondition1} and~\eqref{eq:multivariatetoydesigncondition2}. As the true parameter we randomly choose the value $\true{\parameter} = \inb{27.39, -46.04, -91.81}$.
And concerning the observations, we assume an ideal scenario of noise-free observations 
\begin{align}
\observations = \true{\observations} := \vectorized{\model}\inb{\design, \true{\parameter}},
\end{align}
even though the standard deviation of the normally distributed observation noise is set to $\sigma := 0.1$. 
As is proven in Corollary~\ref{cor:comparison-of-prediction-uncertainty-measures-toy-example-multivariate}, under this ideal scenario of noise-free observations, the cubature formula~\eqref{eq:prediction-uncertainty-lu-darmofal} is exact up to numerical noise. \fig{fig:multivariatetoymodel-validation} shows the approximation error~\eqref{eq:validationloss} for the linearization- and for the cubature-based approximation methods. 
\begin{table}[ht]
    \centering
    \caption{Distribution of the approximaton errors $\Delta_{\design, \observations}^{\mathrm{method}}(\feature)$ for input points $\feature$ ranging in a uniform grid of $N_{\feature} := 10^4$ points in $\featurespace$ and for $\mathrm{method} = \lin$ and $\mathrm{method} = \ludarmofal$. The percentage numbers represent the percentiles of the data in each column.}
    \label{tab:multivariatetoymodel-validation-stats}
    \begin{tabular}{lrr}
\toprule
 & LIN & LD \\
\midrule
mean & 1.50e-07 & 2.67e-13 \\
std & 5.83e-08 & 1.96e-13 \\
min & 8.11e-15 & 9.33e-15 \\
25\% & 1.07e-07 & 1.04e-13 \\
50\% & 1.60e-07 & 2.54e-13 \\
75\% & 1.96e-07 & 3.71e-13 \\
max & 2.37e-07 & 6.91e-13 \\
\bottomrule
\end{tabular}

\end{table}
\tab{tab:multivariatetoymodel-validation-stats}, in turn, exhibits the distribution of the approximation errors~\eqref{eq:validationloss} when the input points $\feature$ are varied in a fine uniform grid in $\featurespace$. The numbers clearly show that for the quadratic model \eqref{eq:multivariatetoymodel}, the cubature-based method outperforms the linearization method. It also confirms numerically the exactness of the Lu-Darmofal approximation, which is rigorously proved in Corollary~\ref{cor:comparison-of-prediction-uncertainty-measures-toy-example-multivariate}.

\subsection{Validation on more general models and designs}\label{sec:benchmarking}
In this section, we benchmark our proposed uncertainty approximation methods on more general models and designs. 
Specifically, we consider three toy models, namely the generic quadratic model \eqref{eq:multivariatetoymodel} in one and two input dimensions and the exponential growth model \eqref{eq:exptoymodel}. Additionally, we consider the NRTL model \eqref{eq:nrtlmodel} as an important example from chemical engineering.
As designs, we consider factorial and equidistant designs for the different models in order to illustrate the design's influence on prediction uncertainty approximation. The used factorial and equidistant design will be specified for each model in the following.
As far as the employed approximation methods are concerned, we compare the uncertainty approximations obtained by linearization (\secn{sec:linearapproximation}) and Lu-Darmofal's cubature method (\secn{sec:ludarmofalapproximation}). The results using McNamee-Stenger's formula (\secn{sec:mcnameestengerapproximation}) are similar to those using Lu-Darmofal's formula, but the latter requires fewer cubature points (higher data efficiency), see~\eqref{eq:N_MS} and~\eqref{eq:N_LD}. Also, and in keeping with the general remarks at the end of \secn{sec:sigmapointapproximation}, Lu-Darmofal's cubature turned out to be superior to the sigma point cubature rule in the specific benchmark examples considered in this section. This is why we only show the results for Lu-Darmofal's cubature (\ludarmofal) and linearization (\lin) here.

In order to assess the quality of these approximation methods for all considered model-design combinations, we use the root mean-squared approximation errors
\begin{equation} \label{eq:predictionuncertaintyerror}
    \Delta^{\mathrm{method}}_{\design, \observations} = \sqrt{\frac{1}{\numberof{\feature}} \sum_{j=1}^{\numberof{\feature}}\inb{V^{\mathrm{method}}_{\design, \observations}\inb{\feature^{(j)}} - V^{\montecarlo}_{\design}\inb{\feature^{(j)}}}^2} ,
\end{equation}
over a fine uniform grid of $\numberof{\feature} := 10^{2\cdot \dimensionof{\feature}}$ input points $\feature^{(j)} \in \featurespace$. In contrast to the local (that is, $\feature$-dependent) approximation error~\eqref{eq:validationloss}, the approximation error~\eqref{eq:predictionuncertaintyerror} is global in the sense that it takes into account all points from a fine uniform grid. Another, even more important, difference compared to~\eqref{eq:validationloss} is that the approximation error~\eqref{eq:predictionuncertaintyerror} uses the Monte-Carlo prediction uncertainty $V_{\design}^{\montecarlo}$ from~\eqref{eq:montecarlovariance} instead of the true prediction uncertainty $V_{\design}$. This is simply because for the more general models and designs considered in this section, we have no analytic expressions anymore for the true prediction uncertainty. The discussion of the results obtained for the different scenarios will be guided by visualizations of
\begin{itemize}
    \item the factorial and the equidistant design $\design$ alongside the true model that we use for noisy observation generation,
    \item the convergence of the Monte Carlo approximator $V^{\montecarlo}_{\design}$ over the number $\numberof{\montecarlo}$ of samples,
    \item the least-squares parameter estimators $\estimator{\parameter}_{\model}\inb{\design, \observations^{(i)}}$ for $\numberof{\montecarlo} := 10^6$ observation samples $\observations^{(1)}, \dots, \observations^{(\numberof{\montecarlo})}$ drawn from~\eqref{eq:random-observation-model}, while for simplicity, only $1000$ random samples are shown in the plots,
    \item the distribution of parameter estimation errors
    \begin{equation} \label{eq:parameterestimatorerror}
        \normtwo{\estimator{\parameter}_{\model}\inb{\design, \observations^{(i)}} - \true{\parameter}},
    \end{equation}
    for the considered observation samples $\observations^{(i)}$ (measured in the Euclidean norm), and
    \item the root mean-squared approximation errors $\Delta^{\mathrm{method}}_{\design, \observations^{(i)}}$ from~\eqref{eq:predictionuncertaintyerror} for the considered observation samples $\observations^{(i)}$.
\end{itemize}
We plot these root mean-squared approximation errors against the parameter estimation errors~\eqref{eq:parameterestimatorerror} in order to reveal how the approximation quality depends on the quality of the parameter estimator. All the arising least-squares estimators $\estimator{\parameter}_{\model}\inb{\design, \observations^{(i)}}$ are determined using a sequential quadratic programming algorithm from the \verb|scipy.optimize| package in \verb|Python| \cite{Virtanen2020}.

In contrast to the works of \cite{Krausch2019}, \cite{Kozachynskyi2024}, we do not use confidence region-based metrics as a quality indicator for the obtained prediction uncertainty estimators. After all, for nonlinear regression models, the predictions based on \leastsquares estimators is no longer normally distributed. As a result, confidence regions as computed in \cite{Krausch2019} are not correct, from a statistical point of view. For practical purposes however, assuming normality of predictions may allow for good results. Still, they are not suitable as a quality indicator in this work, where the estimators by the different methods may be close, and a wrong reference may result in the wrong conclusions.
Instead, we consider the distance of the obtained prediction uncertainty estimations to the \montecarlo-based estimator, and compare the results under consideration of the parameter estimation error.

\subsubsection{Quadratic model in one input dimension}
We begin with the quadratic model~\eqref{eq:multivariatetoymodel} in $\dimensionof{\feature} = 1$ input dimension with unit coefficients~\eqref{eq:model-coefficients-validation}. We will further refer to this model as the \emph{univariate quadratic model}. In Fig. \ref{fig:univariatetoymodel-observationsandmodel}, we observe the two design choices considered in this paper: \textit{factorial} and \textit{equidistant} which are defined, respectively, as
\begin{equation}
\begin{aligned} 
    \design &:= \inb{-1, -1, 1, 1} \\ 
    \design &:= \inb{-1, -0.33, 0.33, 1} .
\end{aligned}
\end{equation}
The true parameter value $\true{\parameter}$ is set to $\true{\parameter} = \inb{2.74 -4.6}$. The standard deviation of the normally distributed observation noise is set to $\std = 0.1$.

\begin{figure}
    \centering
    \includegraphics[width=1.0\textwidth]{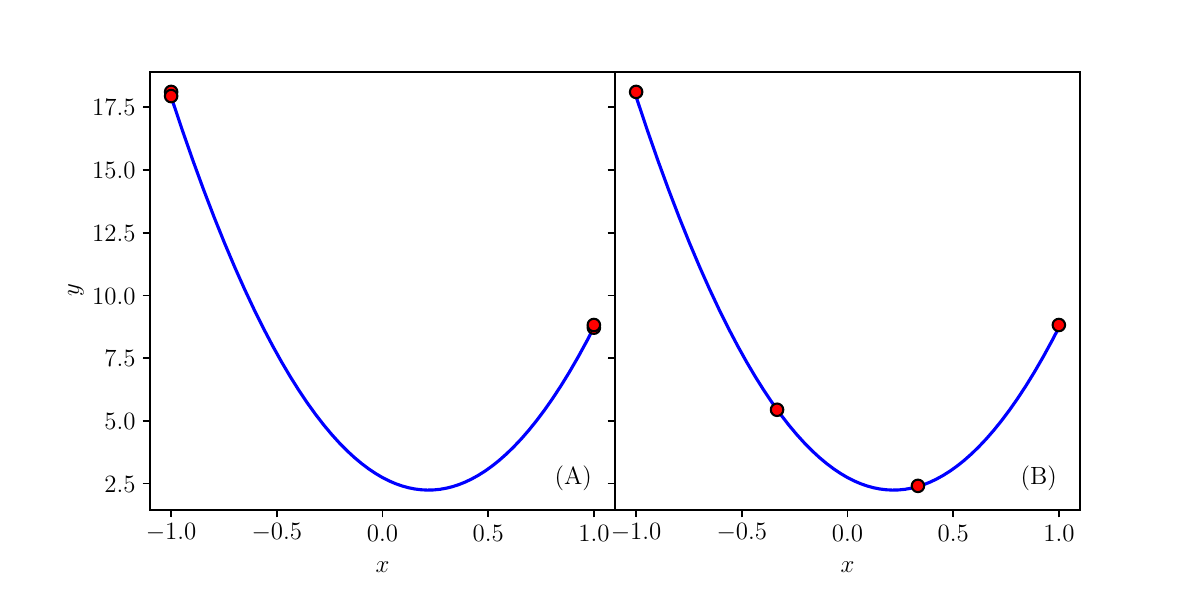}
    \caption{Model and observations of the univariate quadratic model for the factorial (A) and equidistant (B) design. The design points are marked by the red circles and the true model is shown as the blue line. The vertical axes of both subplots (A) and (B) are aligned and have the same scale.}
    \label{fig:univariatetoymodel-observationsandmodel}
\end{figure}

In Fig. \ref{fig:univariatetoymodel-montecarloconvergence}, we see that for both designs, the Monte Carlo estimator for the model's prediction uncertainty converges as we approach the 1 million sample mark.
\begin{figure}
    \centering
    \includegraphics[width=1.0\textwidth]{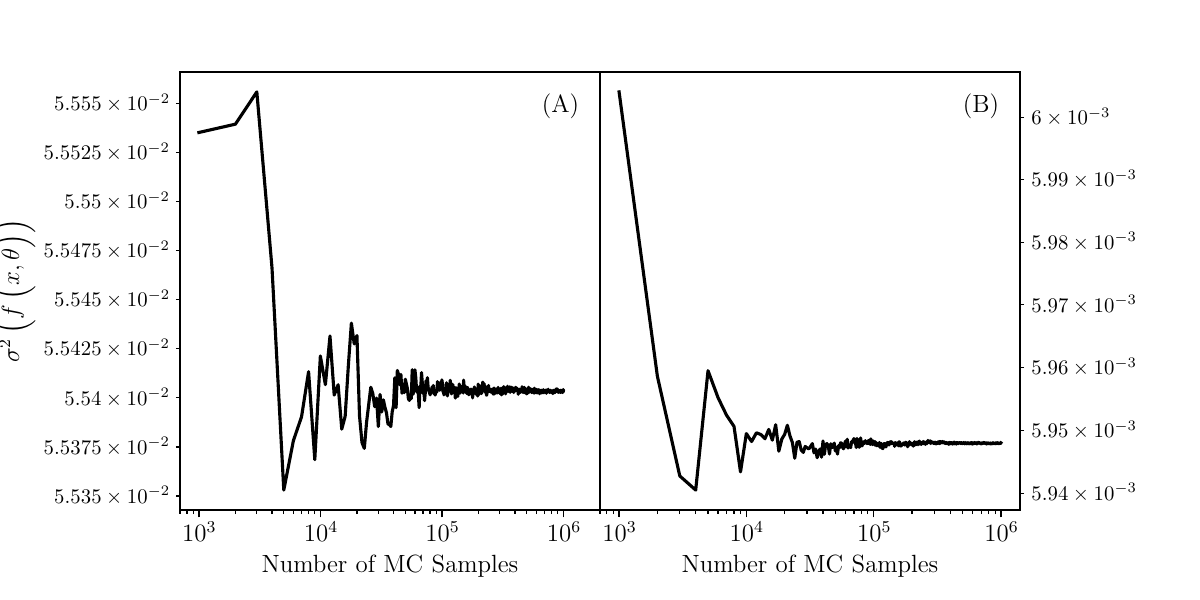}
    \caption{Convergence of the \montecarlo~approximator $V_{\design}^{\montecarlo}(\feature)$ to the prediction uncertainty $V_{\design}(\feature)$ of the univariate quadratic model at $\feature := 0$ over the number of samples for the factorial (A) and equidistant (B) design.}
    \label{fig:univariatetoymodel-montecarloconvergence}
\end{figure}

Fig. \ref{fig:univariatetoymodel-parametersamples} displays the Monte Carlo samples obtained from repeated computer experiments in the parameter space. For both designs, the samples are well-enclosed by an ellipsoid, although the rotation and shape of the ellipsoids differ. Additionally, visual inspection suggests that the true parameter value $\true{\parameter}$ (lime pentagon) and the Monte Carlo mean estimator $\meanof{\parameter}$ (yellow square) coincide. These observations indicate that the linearization may perform well for this model, as it does not exhibit strong nonlinearity with respect to its parameters, as we will further explore with exponential growth model and NRTL model below.
\begin{figure}
    \centering
    \includegraphics[width=1.0\textwidth]{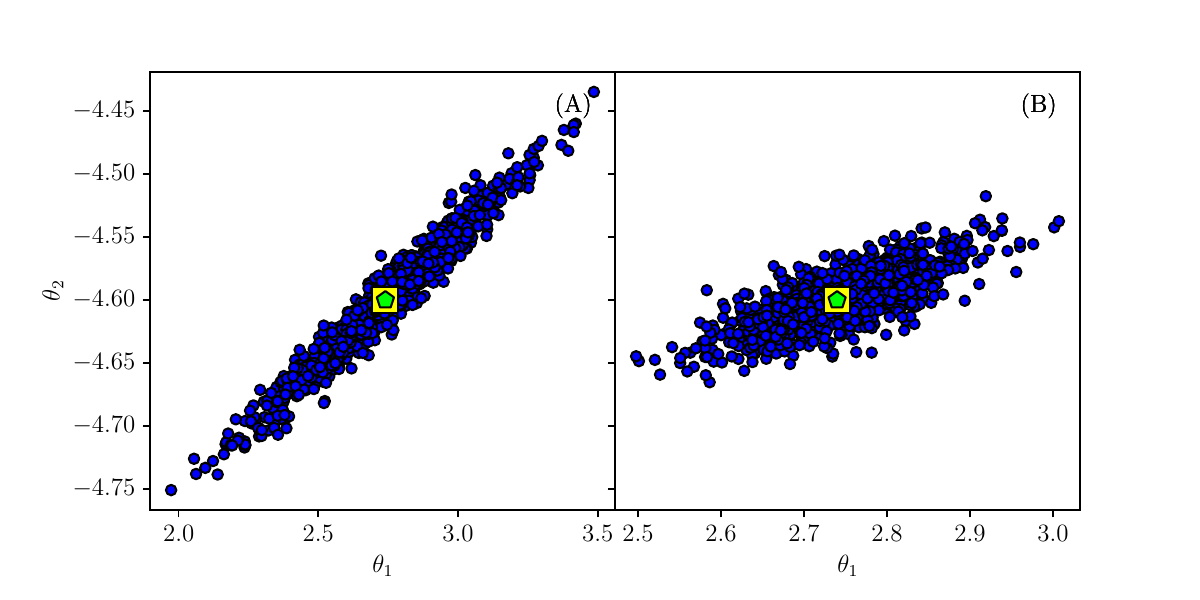}
    \caption{Least-squares estimates for $10^3$ from the $\numberof{\montecarlo} := 10^6$ Monte-Carlo samples of the univariate quadratic model for the factorial (A) and equidistant (B) design. The \montecarlo~samples are marked by the blue circles, the true parameter value $\true{\parameter}$ is marked by the lime pentagon, and the \montecarlo~mean estimator $\meanof{\parameter}$ is marked by the yellow square. The vertical axes of both subplots (A) and (B) are aligned and have the same scale.}
    \label{fig:univariatetoymodel-parametersamples}
\end{figure}

In Fig. \ref{fig:univariatetoymodel-parametererror}, we examine the distribution of parameter estimation errors according to \eqref{eq:parameterestimatorerror}. For both design choices, the majority of parameter samples are in close proximity to the true parameter $\true{\parameter}$, as indicated by the shape of the parameter estimator error distribution, which peaks at a small error value. However, the parameter estimators from the equidistant design are closer to the true parameter value, as evidenced by the scale difference of the horizontal axes in \fig{fig:univariatetoymodel-parametersamples} (A) and (B). 
\begin{figure}
    \centering
    \includegraphics[width=1.0\textwidth]{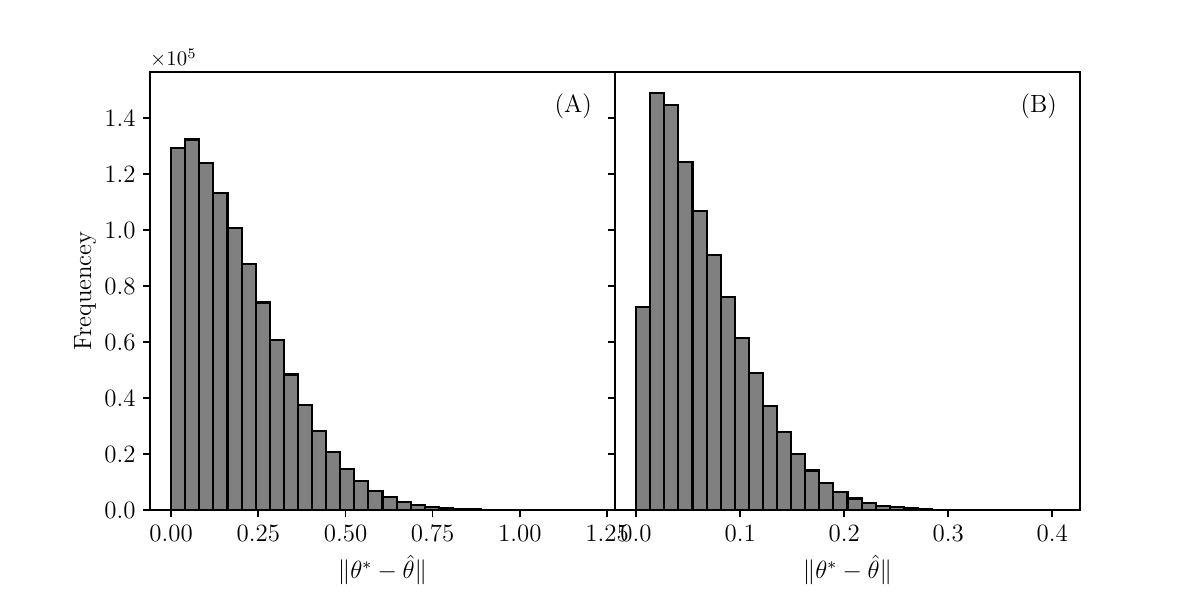}
    \caption{Histrogram of parameter estimation errors \eqref{eq:parameterestimatorerror} of the univariate quadratic model for the factorial (A) and equidistant (B) design. The vertical axes of both subplots (A) and (B) are aligned and have the same scale.}
    \label{fig:univariatetoymodel-parametererror}
\end{figure}

Fig. \ref{fig:univariatetoymodel-distancecomparison} shows the uncertainty estimation error \eqref{eq:predictionuncertaintyerror} versus the parameter estimator error \eqref{eq:parameterestimatorerror} for the factorial (A) and equidistant (B) designs. Both methods, \lin~and \ludarmofal, perform well for small parameter estimator errors, but their performance deteriorates as parameter estimator errors increase. The results are so close that they cannot be easily distinguished through visual observation. The difference in vertical axis scaling in (A) and (B) clearly demonstrates that the equidistant design generally yields better uncertainty estimation than the factorial design, regardless of the estimation method used. As noted earlier, the univariate quadratic model does not exhibit strong nonlinearity with respect to its parameters (cf. \ref{fig:univariatetoymodel-parametersamples}), which explains why the results for the \lin~and \ludarmofal~are so closely aligned.
\begin{figure}
    \centering
    \includegraphics[width=1.0\textwidth]{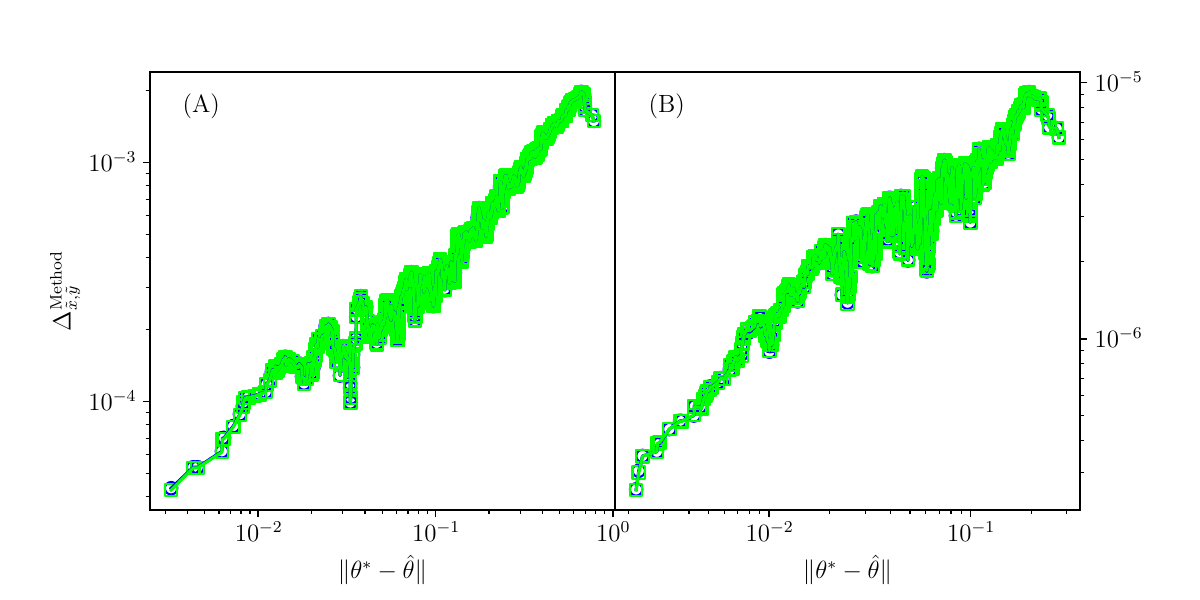}
    \caption{Root mean-squared approximation error \eqref{eq:predictionuncertaintyerror} vs. \leastsquares estimation error \eqref{eq:parameterestimatorerror} for the \montecarlo~samples, obtained for the univariate quadratic model using the factorial (A) and equidistant (B) design. The results for \lin~and \ludarmofal~are shown in blue circles and lime squares, respectively.}
    \label{fig:univariatetoymodel-distancecomparison}
\end{figure}

\subsubsection{Quadratic model in two input dimensions}
We continue with the quadratic model~\eqref{eq:multivariatetoymodel} in $\dimensionof{\feature} = 2$ input dimension with unit coefficients~\eqref{eq:model-coefficients-validation}. We will further refer to this model as the \emph{multivariate quadratic model}. In \fig{fig:multivariatetoymodel-observationsandmodel}, we observe the two design choices considered in this paper: \textit{factorial} and \textit{equidistant} which are defined, respectively, as
\begin{equation}
\begin{aligned} 
    \design &:=
    \begin{pmatrix}
        - 1 & -1 & 1  & 1 & -1 & -1 & 1  & 1 & -1 \\
        - 1 & 1  & -1 & 1 & -1 & 1  & -1 & 1 & -1
    \end{pmatrix}, \\ 
    \design &:=
    \begin{pmatrix}
        - 1 & 0  & 1  & -1 & 0 & 1  & -1  & 0 & 1 \\
        - 1 & -1 & -1 & 0  & 0 & 0  & 1   & 1 & 1
    \end{pmatrix}. 
\end{aligned}
\end{equation}
The true parameter $\true{\parameter} = \inb{27.39, -46.04, -91.81}$ was chosen randomly, and the observation noise is set to $\std = 0.1$.

\begin{figure}
    \centering
    \includegraphics[width=1.0\textwidth]{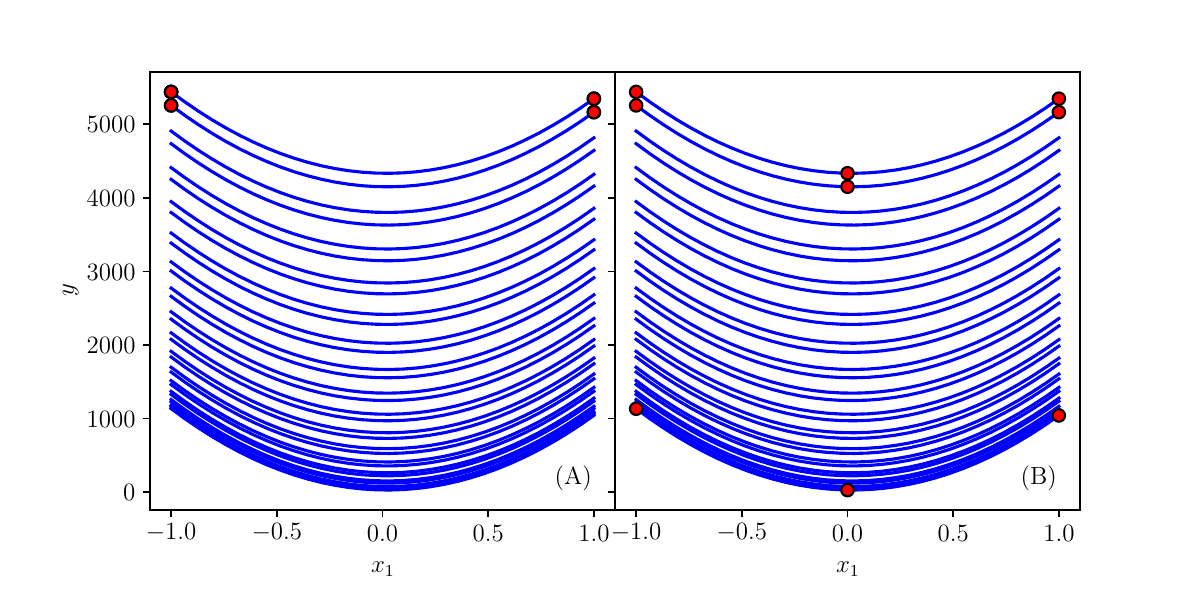}
    \caption{Model and observations of the multivariate quadratic model for the factorial (A) and equidistant (B) design. The design points are marked by the red circles and the true model is shown as the blue lines. The vertical axes of both subplots (A) and (B) are aligned and have the same scale.}
    \label{fig:multivariatetoymodel-observationsandmodel}
\end{figure}

In \fig{fig:multivariatetoymodel-montecarloconvergence}, we see that for both designs, the Monte Carlo estimator for the model's prediction uncertainty converges as we approach the 1 million sample mark.

\begin{figure}
    \centering
    \includegraphics[width=1.0\textwidth]{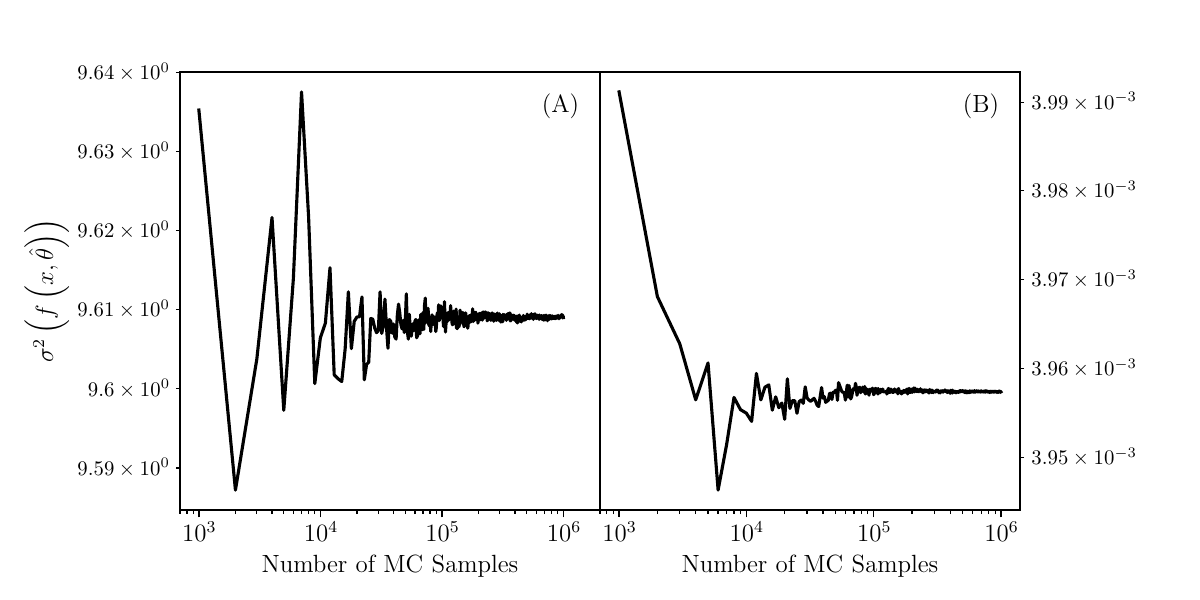}
    \caption{Convergence of the \montecarlo~approximator to the prediction uncertainty of the multivariate quadratic model at $\feature := \inb{-1., 0.03}$ over the number of samples for the factorial (A) and equidistant (B) design.}
    \label{fig:multivariatetoymodel-montecarloconvergence}
\end{figure}

In \fig{fig:multivariatetoymodel-parametersamples}, we examine the Monte Carlo samples obtained from repeated computer experiments in the parameter space. For both designs, most of the samples can be enclosed by an ellipsoid; however, they exhibit notable differences in terms of rotation and size, as indicated by the scaling differences between the horizontal axes in (A) and (B), and the distribution along the vertical axis, which is aligned between (A) and (B).

\begin{figure}
    \centering
    \includegraphics[width=1.0\textwidth]{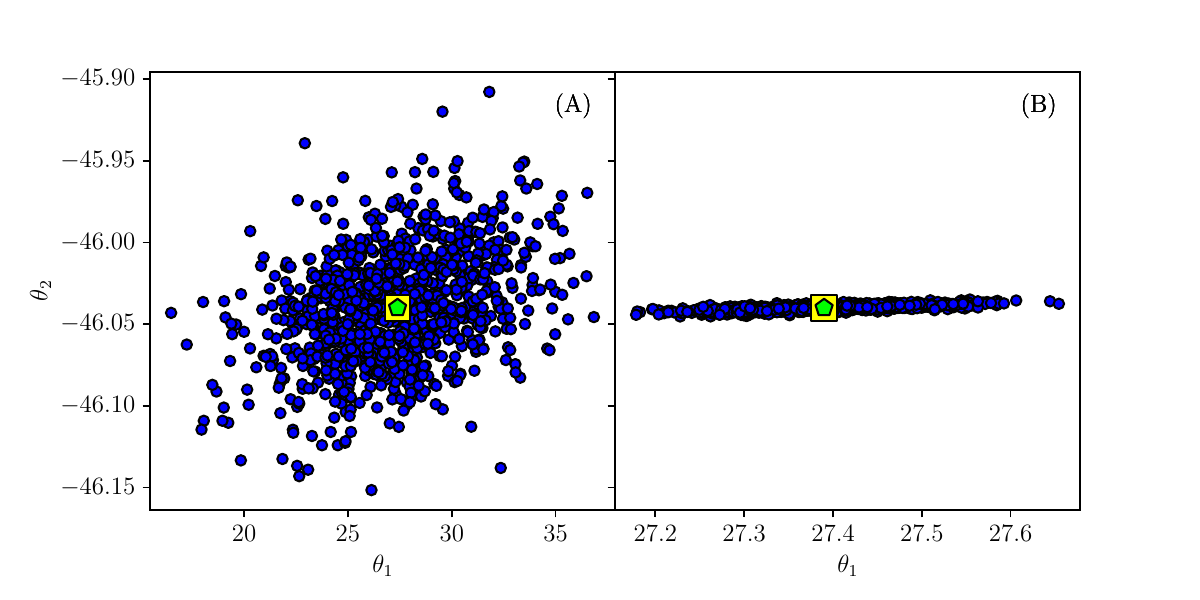}
    \caption{Least-squares estimates for $10^3$ from the $\numberof{\montecarlo} := 10^6$ Monte-Carlo samples of the multivariate quadratic model for the factorial (A) and equidistant (B) design. The \montecarlo~samples are marked by the blue circles, the true parameter value $\true{\parameter}$ is marked by the lime pentagon, and the \montecarlo~mean estimator $\meanof{\parameter}$ is marked by the yellow square. The vertical axes of both subplots (A) and (B) are aligned and have the same scale. Note that while the multivariate quadratic model as considered in this section has three parameters $\parameter_0, \parameter_1, \parameter_2$, only first two of them are shown in the plot for the sake of simplicity. Projections including $\parameter_2$ do not show anything surprising.}
    \label{fig:multivariatetoymodel-parametersamples}
\end{figure}

In \fig{fig:multivariatetoymodel-parametererror}, we see the distribution of parameter estimation errors. For both, factorial (A) and equidistant (B) designs, the majority of parameter estimators is in close vicinity to $\true{\parameter}$ and estimators are increasingly unlikely with increasing distance to $\true{\parameter}$. The difference in axis scaling along the horizontal axis in (A) and (B) again indicates that the equidistant design generally yields better parameter estimators than the factorial design.

\begin{figure}
    \centering
    \includegraphics[width=1.0\textwidth]{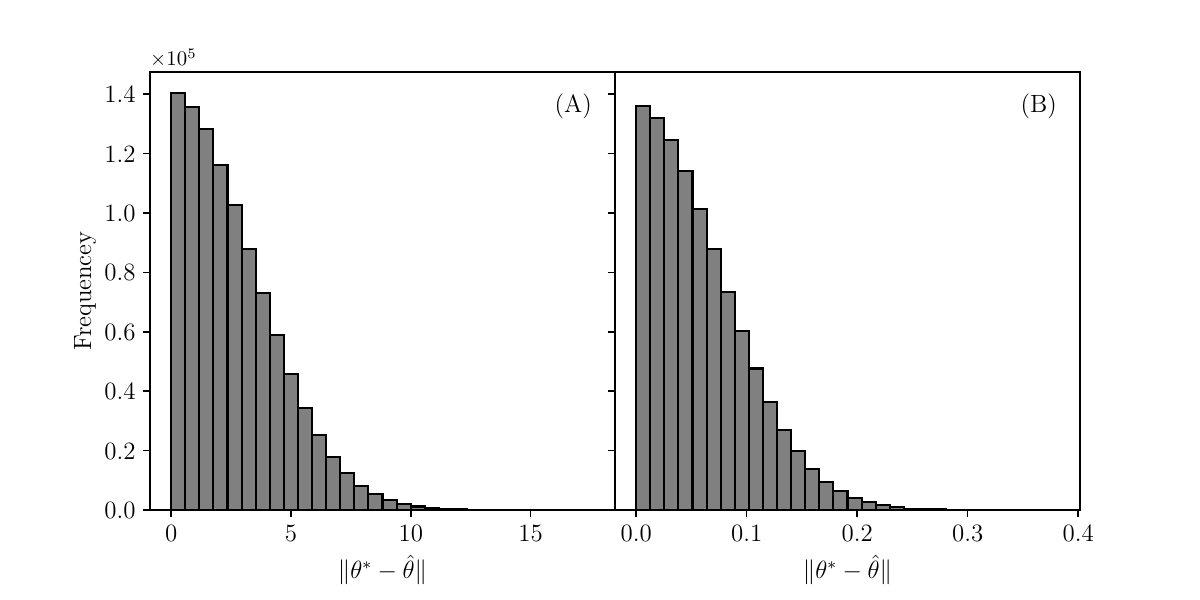}
    \caption{Distribution of parameter estimation errors obtained from the \montecarlo~sampling according to \eqref{eq:parameterestimatorerror} for the multivariate quadratic model at the factorial (A) and equidistant (B) design. The vertical axes of both subplots (A) and (B) are aligned and have the same scale.}
    \label{fig:multivariatetoymodel-parametererror}
\end{figure}

\fig{fig:multivariatetoymodel-distancecomparison} illustrates the uncertainty estimation error \eqref{eq:predictionuncertaintyerror} versus the parameter estimator error \eqref{eq:parameterestimatorerror} for the factorial (A) and equidistant (B) designs. For the factorial design in (A), we obtain a similar picture as for the univariate quadratic model in \fig{fig:univariatetoymodel-distancecomparison} (A). For the equidistant design (B), the results obtained by the two uncertainty estimation methods are very similar as well, while the \lin-based estimators appears to outperform \ludarmofal~for mediocre parameter estimation errors.

\begin{figure}
    \centering
    \includegraphics[width=1.0\textwidth]{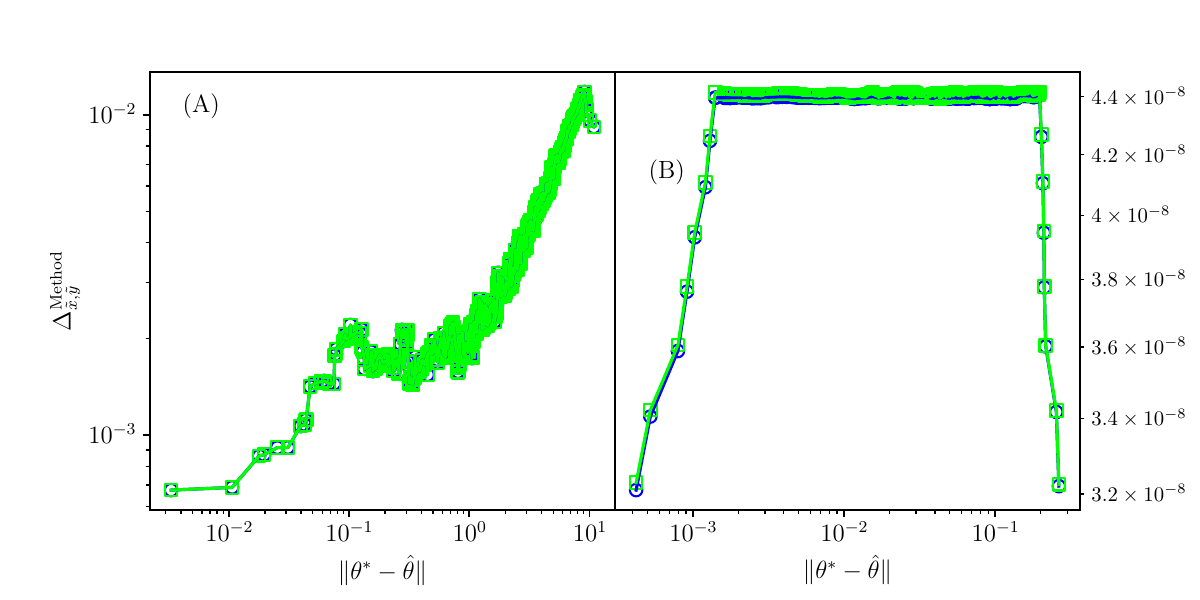}
    \caption{Least-squares estimator error \eqref{eq:parameterestimatorerror} from \montecarlo~sampling are visualized at their respective variance approximation error \eqref{eq:predictionuncertaintyerror} for the multivariate quadratic model using the factorial (A) and equidistant (B) design. The results for \lin~are shown in blue and those for \ludarmofal~in lime.}
    \label{fig:multivariatetoymodel-distancecomparison}
\end{figure}

\subsubsection{Exponential growth model}
After having extensively discussed the quadratic model, we now move on to a more nonlinear model, namely the \textit{exponential growth model}. It is defined by
\begin{equation} \label{eq:exptoymodel}
    \model\inb{\feature, \parameter} := \parameter_1 \cdot \exp\inb{\theta_2 \cdot x} \quad \inb{\inb{\feature, \parameter} \in \featurespace \times \parameterspace} .
\end{equation}
and arises as a building block in many growth models, whence it is widely studied in the statistics literature, see~\cite{Bates1988, Seber2003}, for instance. Depending on the choice of the true parameter value $\true{\parameter}$, the confidence region of the least-squares estimator in the parameter space can become highly non-elliptic. 
In \fig{fig:exptoymodel-observationsandmodel}, we see the two design choices considered in this paper: \textit{factorial} and \textit{equidistant} which are defined, respectively, as
\begin{equation}
\begin{aligned}
    \design &:= \inb{-1, -1, 1, 1} \\
    \design &:= \inb{-1, -0.33, 0.33, 1} . 
\end{aligned}
\end{equation}
The true parameter value is set to $\true{\parameter} = \inb{0.2, 1.2}$, and the standard deviation of the normally distributed observation noise is set to $\std = 0.1$.

In \fig{fig:exptoymodel-observationsandmodel} (A), we observe that one of the observations has a negative $y$-value. This is theoretically not possible given the formulation of exponential growth model. However, due to the observation noise, negative values can occur. This is more likely in the factorial design, where there are two measurements at $x = -1$, compared to the equidistant design, which has only one measurement at that point. Thus, the observations at $x = -1$ are more likely to fall into the negative range, as the function value is smallest there.

\begin{figure}
    \centering
    \includegraphics[width=1.0\textwidth]{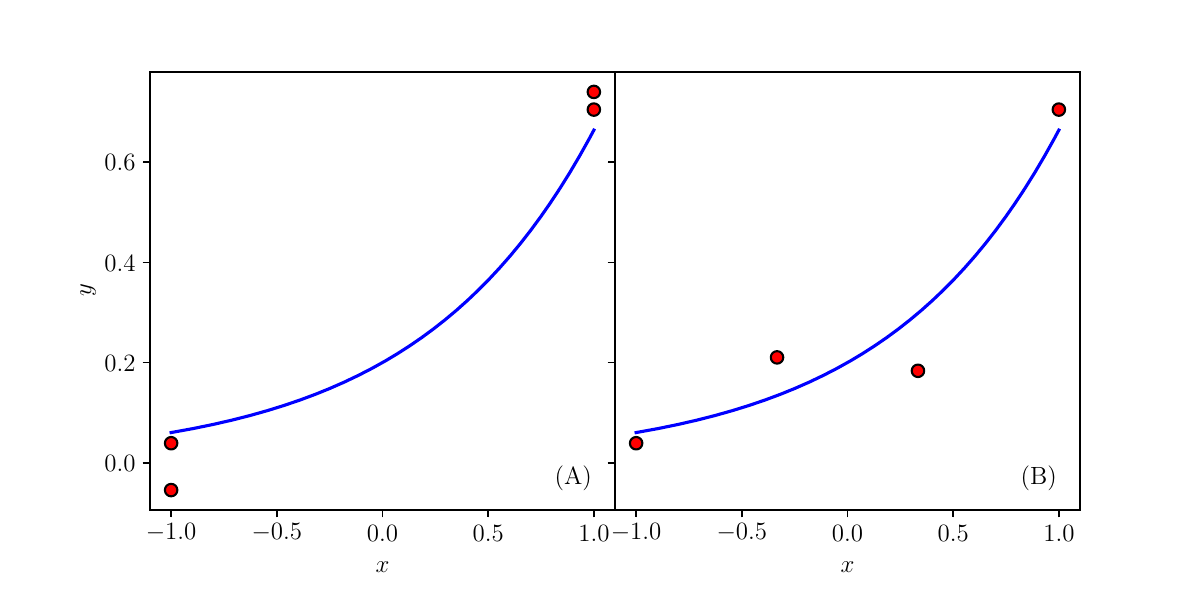}
    \caption{The blue lines show the exponential growth model at the evaluation points (same for (A) and (B)) and the red points show the factorial (A) and equidistant (B) design. The vertical axes of both subplots (A) and (B) are aligned and have the same scale.}
    \label{fig:exptoymodel-observationsandmodel}
\end{figure}

In \fig{fig:exptoymodel-montecarloconvergence}, we see that for both designs, the Monte Carlo estimator for the model's prediction uncertainty converges as we approach the 1 million sample mark.

\begin{figure}
    \centering
    \includegraphics[width=1.0\textwidth]{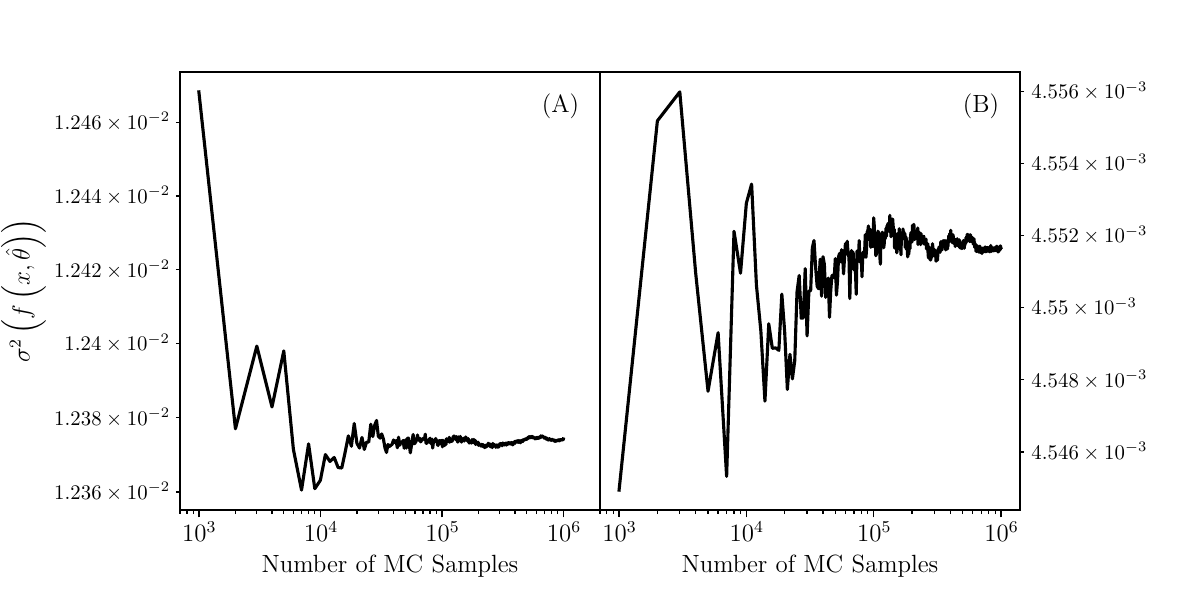}
    \caption{Convergence of the \montecarlo~approximator to the prediction uncertainty of the exponential growth model at $\feature := 0$ over the number of samples for the factorial (A) and equidistant (B) design.}
    \label{fig:exptoymodel-montecarloconvergence}
\end{figure}

In \fig{fig:exptoymodel-parametersamples}, we examine the Monte Carlo samples obtained from repeated computer experiments in the parameter space. For both designs, the samples do not exhibit an elliptic distribution. This suggests that we can expect poor results from the \lin-based uncertainty approximation when comparing uncertainty estimators below. Also, the mean of the Monte Carlo parameter estimators does not provide a good approximation of the true parameter value due to the non-elliptic and non-symmetric distribution of samples.

\begin{figure}
    \centering
    \includegraphics[width=1.0\textwidth]{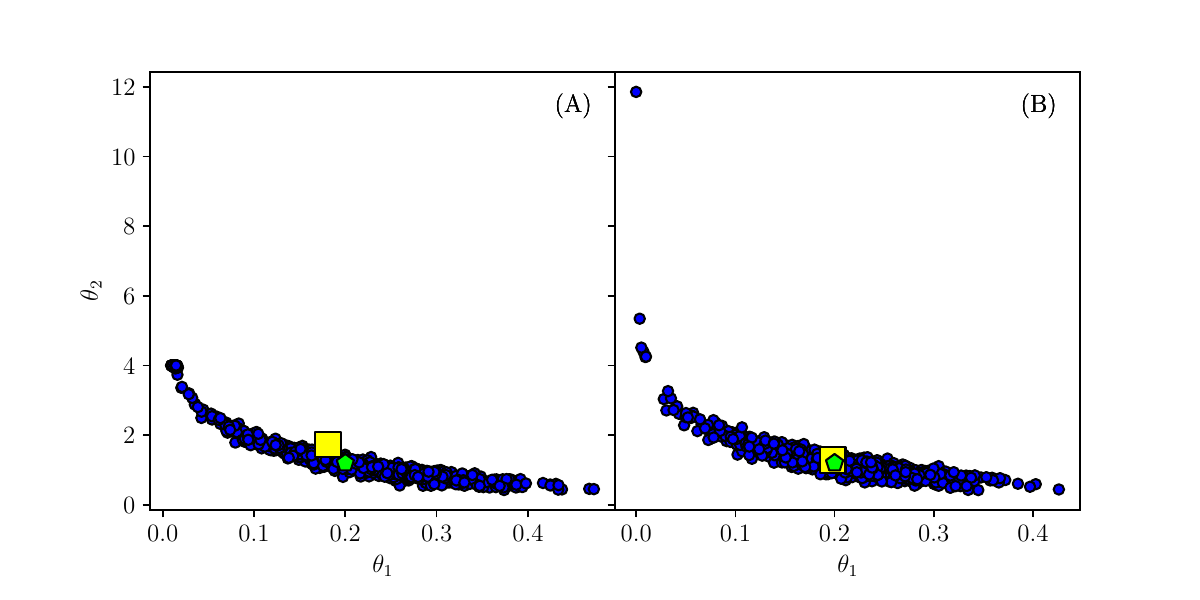}
    \caption{Least-squares estimates for $10^3$ from the $\numberof{\montecarlo} := 10^6$ Monte-Carlo samples of the exponential growth model for the factorial (A) and equidistant (B) design. The \montecarlo~samples are marked by the blue circles, the true parameter value $\true{\parameter}$ is marked by the lime pentagon, and the \montecarlo~mean estimator $\meanof{\parameter}$ is marked by the yellow square. The vertical axes of both subplots (A) and (B) are aligned and have the same scale.}
    \label{fig:exptoymodel-parametersamples}
\end{figure}

In \fig{fig:exptoymodel-parametererror}, we see the distribution of parameter estimation errors, which are calculated based on the Euclidean error between the least-squares estimator obtained for the respective Monte Carlo sample and the true parameter value, assumed to be known for comparison. 

\begin{figure}
    \centering
    \includegraphics[width=1.0\textwidth]{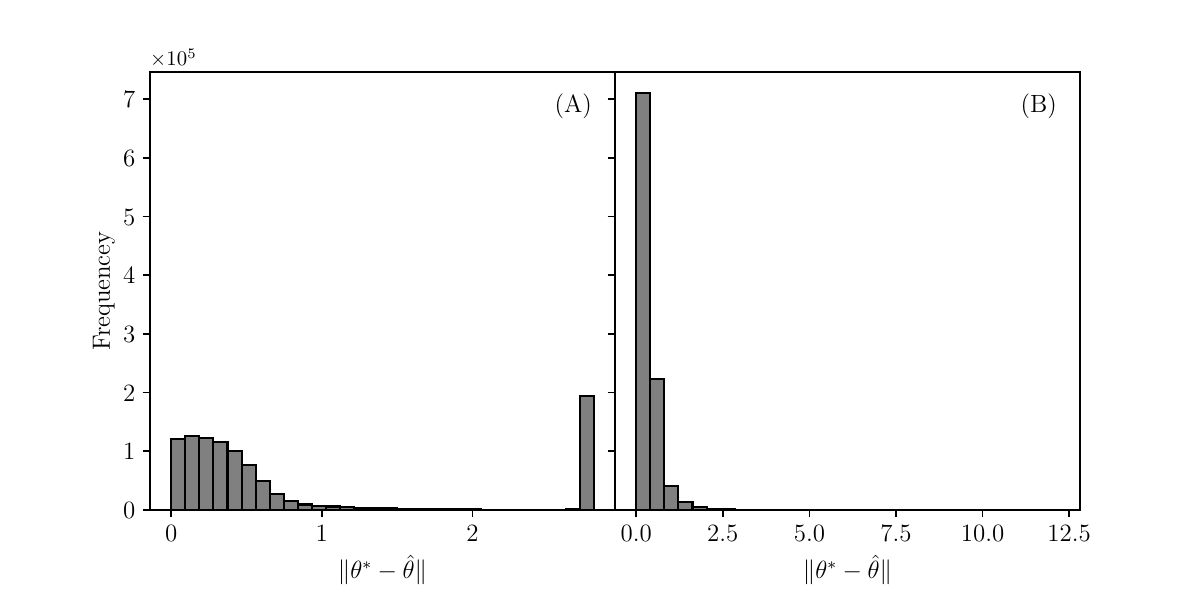}
    \caption{Distribution of parameter estimation errors obtained from the \montecarlo~sampling according to \eqref{eq:parameterestimatorerror} for the exponential growth model at the factorial (A) and equidistant (B) design. The vertical axes of both subplots (A) and (B) are aligned and have the same scale.}
    \label{fig:exptoymodel-parametererror}
\end{figure}

\fig{fig:exptoymodel-distancecomparison} shows the uncertainty estimation error \eqref{eq:predictionuncertaintyerror} versus the parameter estimator error \eqref{eq:parameterestimatorerror} for the factorial (A) and equidistant (B) designs. For the factorial design (A), the standard deviation estimator by \ludarmofal~is closer to the Monte Carlo approximator than that of the \lin-method, especially for high parameter estimator errors, where the uncertainty estimation error from the linearization tends to worsen while the cubature-based estimator remains robust. In contrast, for the equidistant design, there exists a range of mediocre parameter estimator error where the uncertainty approximation by the linearization is superior to that of the cubature. However, this range is limited and does not account for the majority of samples. This observation illustrates that due to the nonlinearities of the parameter estimation problem, a moderate parameter estimator can still yield a surprisingly good uncertainty estimator (B) or become arbitrarily bad (A). For the factorial design (A), the situation is more critical, as the linearization is more sensitive to bad least-squares estimators.

\begin{figure}
    \centering
    \includegraphics[width=1.0\textwidth]{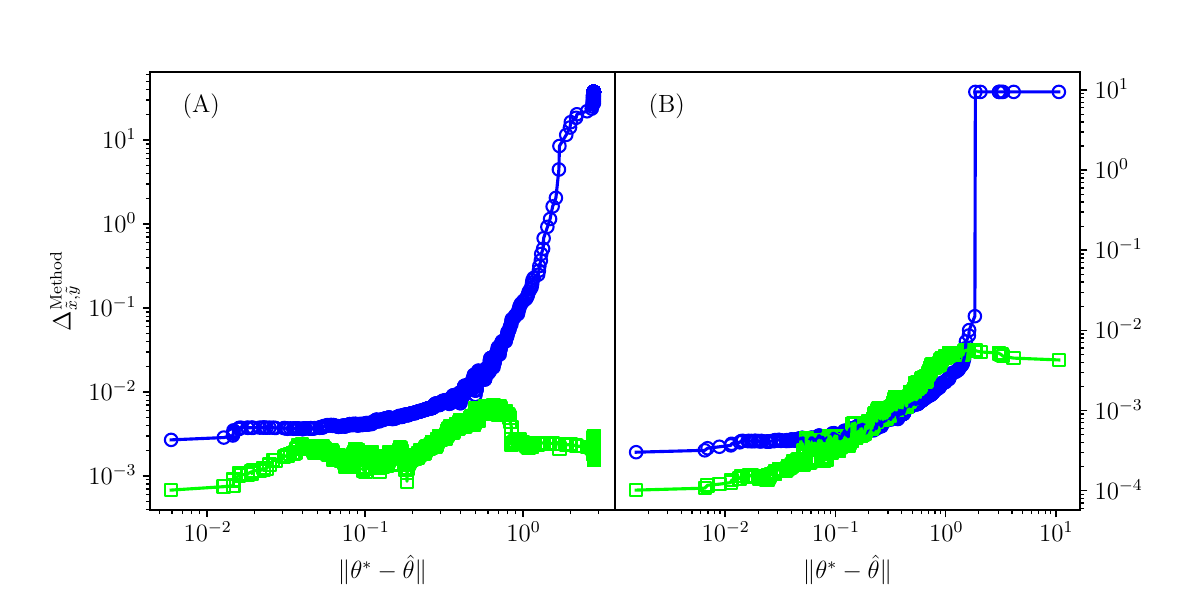}
    \caption{Least-squares estimator error \eqref{eq:parameterestimatorerror} from \montecarlo~sampling are visualized at their respective variance approximation error \eqref{eq:predictionuncertaintyerror} for the multivariate quadratic model using the factorial (A) and equidistant (B) design. The results for \lin~are shown in blue and those for \ludarmofal~in lime.}
    \label{fig:exptoymodel-distancecomparison}
\end{figure}

This emphasizes that the cubature-based uncertainty estimation can serve as a robust alternative to linearization, especially for models with highly non-elliptic parameter space confidence regions.

\subsubsection{NRTL model}
As our last benchmark model, we consider the \emph{non-random two-liquid model} (NRTL model) for the activity coefficient 
\begin{align}
f(l,T,\parameter) := \gamma_1(l,T,\theta)
\end{align}
of the first component of a binary mixture. According to \cite{Renon1968}, this model is defined by
\begin{align} \label{eq:nrtlmodel}
    &\gamma_i(l, T, \parameter) 
    := \exp\bigg( \frac{\sum_{k=1}^2 l_k \tau_{ki}(T, \parameter) G_{ki}(T, \parameter)}{\sum_{k=1}^2 l_k G_{ki}(T, \parameter)} \notag\\
    &\qquad + \sum_{j=1}^2 \frac{l_j G_{ij}(T, \parameter)}{\sum_{k=1}^2 l_k G_{kj}(T, \parameter)} \bigg( \tau_{ij}(T, \parameter) - \frac{\sum_{k=1}^2 l_k \tau_{kj}(T, \parameter) G_{kj}(T, \parameter)}{\sum_{k=1}^2 l_k G_{kj}(T, \parameter)}  \bigg) \bigg),
\end{align}
where $(l_1,l_2) := (l,1-l)$ denote the liquid-phase molar fractions of the two components, $T$ denotes the temperature, and $G_{ij}$ and $\tau_{ij}$ are defined by
\begin{gather}
    G_{ij}(T, \parameter) := \exp(-\alpha_{ij}(T)\tau_{ij}(T, \parameter)) , \label{eq:NRTL-G}\\
    \tau_{ij}(T, \parameter) := a_{ij} + \frac{b_{ij}}{T} \qquad (i\ne j) \qquad \text{and} \qquad \tau_{ii}(T) := 0,  \label{eq:NRTL-tau}\\
    \alpha_{ij}(T, \parameter) := c_{ij} + d_{ij} T \qquad (i\ne j) \qquad \text{and} \qquad \alpha_{ii}(T) := 0  \label{eq:NRTL-alpha}
\end{gather}
for $i\in\{1,2\}$. As in the original model proposed in \cite{Renon1968}, we assume the non-randomness terms $\alpha_{ij}$ to be symmetric and temperature-independent, that is,
\begin{align} \label{eq:NRTL-alpha-symmetric-and-T-independent}
    c_{ij} = c_{ji} \qquad \text{and} \qquad d_{ij} = 0 \qquad (i,j\in\{1,2\}).
\end{align}
Furthermore, we fix the parameters as follows:
\begin{align}
    a_{12} = 0, \qquad a_{21} = 0, \qquad 
    c_{12} = c_{21} = 0.3,
\end{align}
so that the adjustable parameters in this benchmark become $\parameter = \inb{b_{12}, b_{21}}$. The true parameter value is set to $\true{\parameter} = \inb{-173.4982, -61.8175}$, and the standard deviation of the normally distributed observation noise is set to $\std = 0.1$.

In \fig{fig:nrtl-observationsandmodel}, we see the two design choices considered in this paper: \textit{factorial} and \textit{equidistant} which are defined, respectively, as
\begin{equation}
\begin{aligned} 
    \design &:=
            \begin{pmatrix}
                0.01      & 0.01      & 0.99      & 0.99      & 0.01      & 0.01      & 0.99      & 0.99      & 0.01 \\
                T_1 & T_3 & T_1 & T_3 & T_1 & T_3 & T_1 & T_3 & T_1
            \end{pmatrix},  \\
    \design &:=
        	    \begin{pmatrix}
                    0.01      & 0.5    & 0.99      & 0.01      & 0.5      & 0.99      & 0.01      & 0.5    & 0.99 \\
                T_1 & T_1 & T_1 & T_2 & T_2   & T_2 & T_3 & T_3 & T_3
            \end{pmatrix}, 
\end{aligned}
\end{equation}
with $T_1 := 298.15$, $T_2 := 335.15$, and $T_3 := 373.15$. 
\begin{figure}
    \centering
    \includegraphics[width=1.0\textwidth]{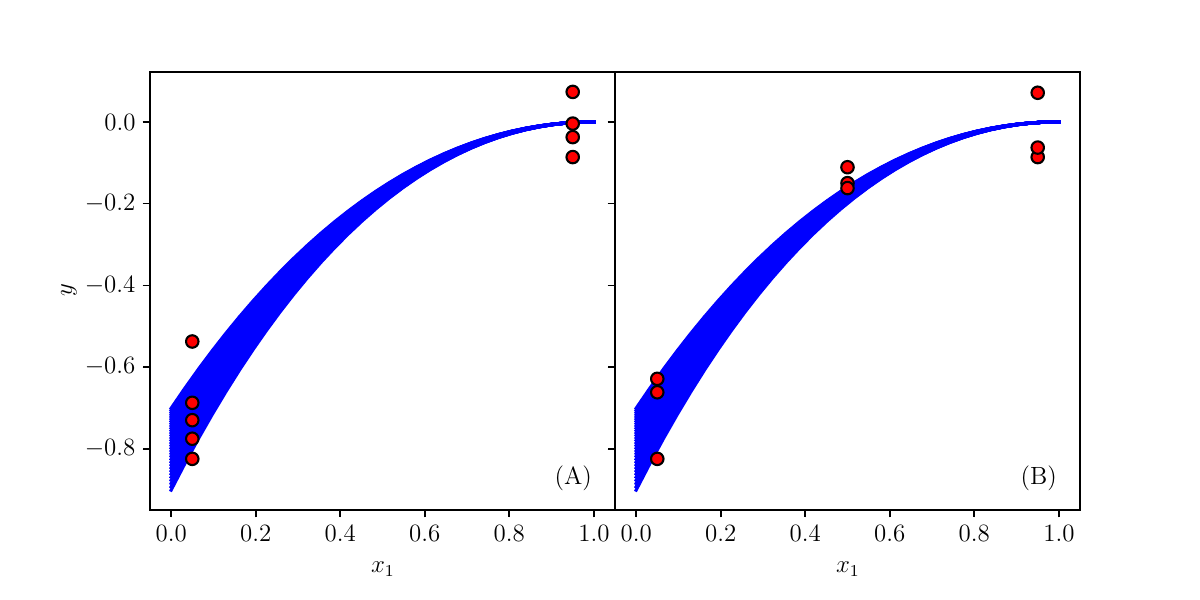}
    \caption{The blue lines show the NRTL model at the evaluation points (same for (A) and (B)) and the red points show the factorial (A) and equidistant (B) design. The vertical axes of both subplots (A) and (B) are aligned and have the same scale.}
    \label{fig:nrtl-observationsandmodel}
\end{figure}

In \fig{fig:nrtl-montecarloconvergence}, we see that for both designs, the \montecarlo~estimator for the model's prediction uncertainty converges as we approach the 1 million sample mark. However, convergence appears to be faster for the factorial design (A), similar to the reults obtained for the exponential growth model (c.f. \fig{fig:exptoymodel-montecarloconvergence}).

\begin{figure}
    \centering
    \includegraphics[width=1.0\textwidth]{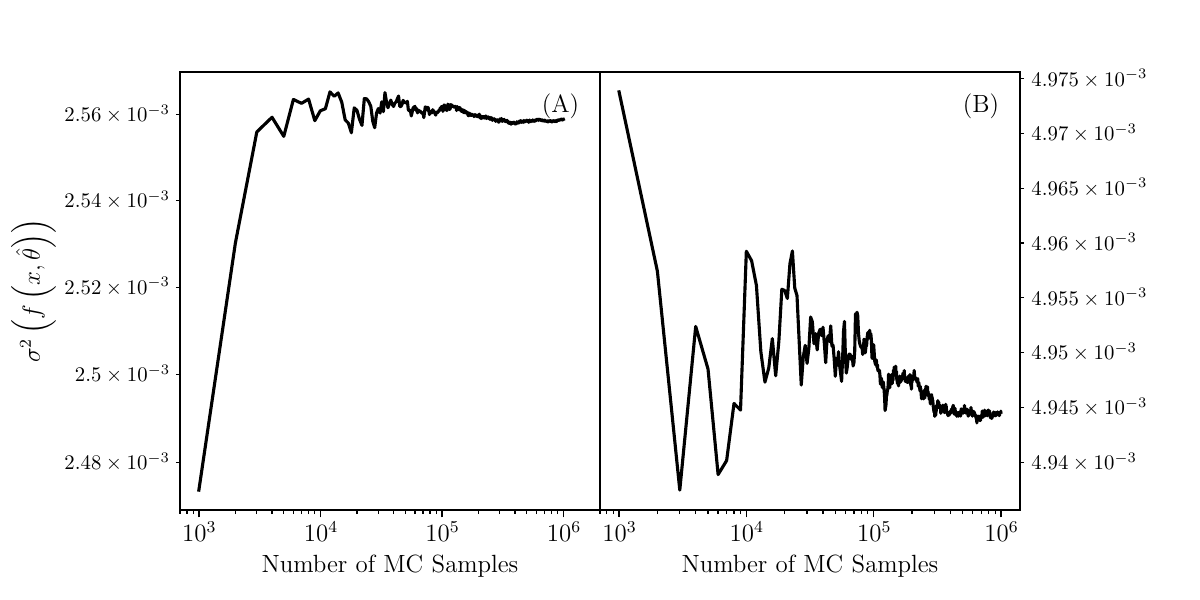}
    \caption{Convergence of the \montecarlo~approximator to the prediction uncertainty of the NRTL model at $\feature := \inb{0, 336.86}$ over the number of samples for the factorial (A) and equidistant (B) design.}
    \label{fig:nrtl-montecarloconvergence}
\end{figure}

In \fig{fig:nrtl-parametersamples}, we examine the \montecarlo~samples obtained from repeated computer experiments in the parameter space. For both designs, the parameter estimator samples from the \montecarlo~observation sets do exhibit an elliptic distribution, with the effect appearing stronger for the equidistant design in (B). This raises the expectation that the linearization-based uncertainty estimation will underperform the cubature-based method when comparing uncertainty estimators below.

\begin{figure}
    \centering
    \includegraphics[width=1.0\textwidth]{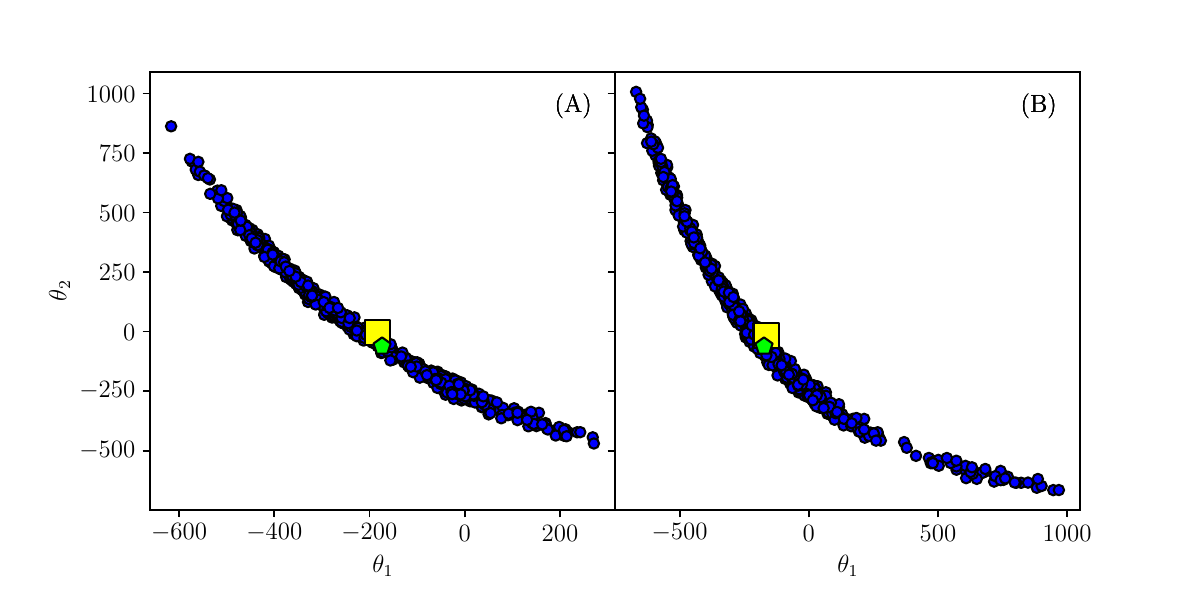}
    \caption{Least-squares estimates for $10^3$ from the $\numberof{\montecarlo} := 10^6$ Monte-Carlo samples of the NRTL model for the factorial (A) and equidistant (B) design. The \montecarlo~samples are marked by the blue circles, the true parameter value $\true{\parameter}$ is marked by the lime pentagon, and the \montecarlo~mean estimator $\meanof{\parameter}$ is marked by the yellow square. The vertical axes of both subplots (A) and (B) are aligned and have the same scale.}
    \label{fig:nrtl-parametersamples}
\end{figure}

In \fig{fig:nrtl-parametererror}, we see the distribution of parameter estimation errors, calculated based on the Euclidean distance between the \leastsquares estimator obtained for the respective \montecarlo~sample and the true parameter value, which we assume to be known for comparison.

\begin{figure}
    \centering
    \includegraphics[width=1.0\textwidth]{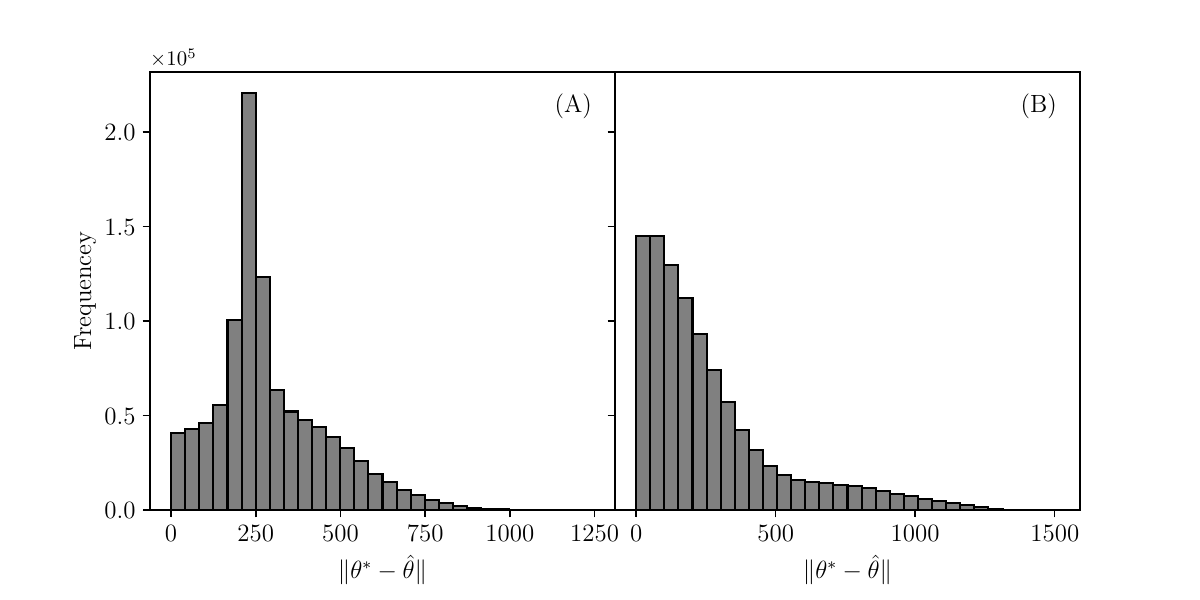}
    \caption{Distribution of parameter estimation errors obtained from the \montecarlo~sampling according to \eqref{eq:parameterestimatorerror} for the NRTL model at the factorial (A) and equidistant (B) design. The vertical axes of both subplots (A) and (B) are aligned and have the same scale.}
    \label{fig:nrtl-parametererror}
\end{figure}

\fig{fig:nrtl-distancecomparison} shows the uncertainty estimation error \eqref{eq:predictionuncertaintyerror} versus the parameter estimator error \eqref{eq:parameterestimatorerror} for the factorial (A) and equidistant (B) designs. For both designs, the prediction uncertainty estimator obtained by the cubature is superior to that obtained by the linearization. Especially for good parameter estimates, the cubature error is particularly small, while it starts significantly higher for the linearization-based estimator, as seen in \ref{fig:nrtl-distancecomparison} (A and B). Additionally, as observed in the other validation scenarios in this section, the results obtained from the equidistant design are superior to those from the factorial design when using the \lin~method, while the results obtained for the cubature-based uncertainty estimator perform similar for both designs.

\begin{figure}
    \centering
    \includegraphics[width=1.0\textwidth]{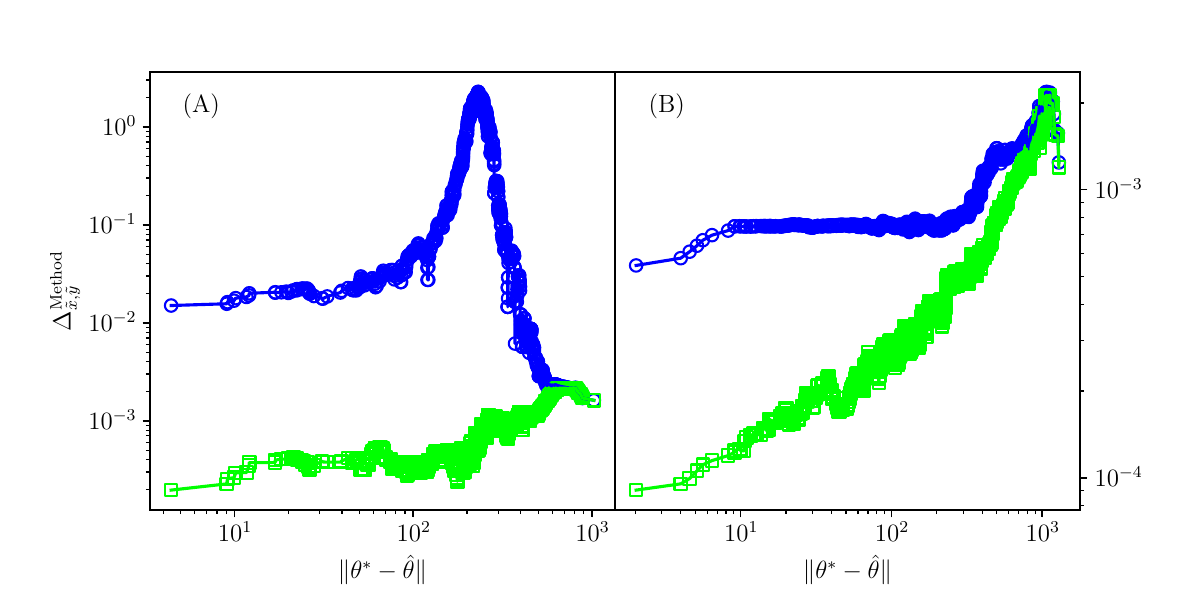}
    \caption{Least-squares estimator error \eqref{eq:parameterestimatorerror} from \montecarlo~sampling are visualized at their respective variance approximation error \eqref{eq:predictionuncertaintyerror} for the NRTL model using the factorial (A) and equidistant (B) design. The results for \lin~are shown in blue and those for \ludarmofal~in lime.}
    \label{fig:nrtl-distancecomparison}
\end{figure}

The results observed for the NRTL model with factorial designs align with those by Kozachynskyi et al. \cite{Kozachynskyi2024}, who note that the linearization-based uncertainty approximation can perform fairly well for the NRTL model. Comparing the distribution of parameter estimator samples obtained for the exponential growth model (cf. \fig{fig:exptoymodel-parametersamples}) to those for the NRTL model (cf. \ref{fig:nrtl-parametersamples}), we find that the samples for the exponential growth model are more non-elliptic. Nonetheless, regardless of the ellipticity of the parameter estimator samples, we find that the cubature-based approximator for prediction uncertainty consistently outperforms the linearization-based approach for both models.

\section{Conclusion and outlook} \label{sec:conclusionandoutlook}
In this work, we have proposed new methods for estimating the prediction uncertainty of nonlinear regression models using cubature formulas. We derived exact expressions for the variance and parameter estimator of a quadratic nonlinear regression model for factorial designs. The proposed cubature-based uncertainty estimation methods were validated on a variety of case studies, including the NRTL model from the chemical engineering domain. From the numerical results obtained during validation, we draw the following conclusions:
\begin{itemize}
    \item As is known for nonlinear regression models, the experimental design matters \cite{Fedorov2014}. The comparison of factorial and equidistant designs shows that for the generic quadratic model in one- and two input dimensions, exponential growth model, and the NRTL model, the equidistant design generally yields better uncertainty estimation than the factorial design, regardless of the estimation method used, as can be seen in \fig{fig:univariatetoymodel-distancecomparison}, \fig{fig:multivariatetoymodel-distancecomparison}, \ref{fig:exptoymodel-distancecomparison}, and \ref{fig:nrtl-distancecomparison}.
    \item For nonlinear regression models, the quality of the parameter estimator strongly affects the prediction uncertainty approximation. While this may be obvious, the results obtained in \secn{sec:benchmarking} show that the progression of approximation error with parameter estimation error is model-dependent. Even though the results in \fig{fig:univariatetoymodel-distancecomparison}, \fig{fig:multivariatetoymodel-distancecomparison}, \ref{fig:exptoymodel-distancecomparison}, and \ref{fig:nrtl-distancecomparison} demonstrate that the cubature-based approximator to the prediction uncertainty is generally superior to the linearization-based approach, and that it is more robust to poor parameter estimators within the models compared in this study, there is no general rule that can be derived from the results. In practice, if the parameter estimator is known to be poor, the uncertainty approximation may also be poor, and thus, should be used with caution, regardless of the approximation method used.
    \item Finally, we find that even though the prediction uncertainty estimator that relies on the linearization, which does not describe the nonlinearity in the model, still offers a fair approximation to the prediction uncertainty. Especially for models that exhibit weak parameter nonlinearity, e.g., the generic quadratic model in one input dimension in \fig{fig:univariatetoymodel-distancecomparison}, the linearization-based method may be a good choice, as it is computationally less expensive than the cubature-based methods. However, the cubature-based methods provide a more robust approximation to the prediction uncertainty, as they are less sensitive to the quality of the parameter estimator, cf. \fig{fig:univariatetoymodel-distancecomparison}, \ref{fig:exptoymodel-distancecomparison}, \ref{fig:nrtl-distancecomparison}.
\end{itemize}

The main limitation of the proposed cubature-based uncertainty approximation methods is that they require many parameter estimations, which may be expensive or even intractable in practice. Also, the number of parameter estimators required rises quadratically with the number of measurement locations, which limits the usability of the method to small-data studies. Future work and improvements on the method could thus address the computational effort of the method as well as the number of parameter estimators necessary to extend its applicability to larger data sets, and by that, a wider range of models.
Using even more accurate formulas for the integrals than the cubature formulas contained in this work may not be promising. We find that the limitation of the approaches comes from the method of generating perturbations of the observed data and generating uncertainty approximations on those perturbations. While that works fine for linear regression models, for nonlinear regression models, it simply results in a bias to the true uncertainty as the approximation does not reflect the actual process, but that of the parameter estimator available at the time. For nonlinear regression models, the confidence region of model predictions exhibits a strong influence on the parameter value, which limits the accuracy that can be obtained by the proposed methods.

\section{Closed-form expressions for the generic quad-ratic model} \label{sec:appendix}
In this appendix, we establish the closed-form expressions that we use for the validation of our cubature-based approximation methods. Specifically, we derive closed-form expressions for the prediction uncertainty of the generic quadratic model~\eqref{eq:multivariatetoymodel} trained on the specific factorial design~\eqref{eq:multivariatetoydesigncondition1} and~\eqref{eq:multivariatetoydesigncondition2}, and for the various approximations to the prediction uncertainty considered in this paper. We point out that all the following results are valid generically, that is, for arbitrary input space dimension $\dimx \in \N$ and for arbitrary choices of the linear and quadratic coefficients
\begin{align}
\alpha_k \in \R \setminus \{0\} 
\qquad \text{and} \qquad
\beta_k \in \R 
\qquad (k \in \{1, \dots, \dimx\})
\end{align}
of our model~\eqref{eq:multivariatetoymodel}. In the following, the index $i$ always refers to the $i$th experiment or observation, whereas the index $k$ refers to the $k$th component of an input point $x \in \featurespace$ or a model parameter $\theta \in \parameterspace$. In particular, $x_{i,k}$ is the $k$th component of the $i$th experiment $x_i = (x_{i,1}, \dots, x_{i,\dimx}) \in \featurespace$, whereas $x_k$ is the $k$th component of an arbitrary input point $x = (x_1, \dots, x_{\dimx}) \in \featurespace$. 

\subsection{Closed-form expression for the prediction uncertainty}

We begin by establishing a closed-form expression for the least-squares estimator of the quadratic model~\eqref{eq:multivariatetoymodel} on the specific factorial design~\eqref{eq:multivariatetoydesigncondition1} and~\eqref{eq:multivariatetoydesigncondition2}.

\begin{lm} \label{lm:least-squares-estimator-toy-example-multivariate}
Suppose $f$ is the generic quadratic model defined by~\eqref{eq:multivariatetoymodel} and $\tilde{x} = (x_1,\dots,x_n) \in \mathcal{X}$ is an experimental design satisfying~\eqref{eq:multivariatetoydesigncondition1} and \eqref{eq:multivariatetoydesigncondition2}. Then for every $\tilde{y} = (y_1,\dots,y_n) \in \Realspace^n$, there exists a unique least-squares estimator $\hat{\theta}_f(\tilde{x},\tilde{y}) = (\hat{\theta}_k)_{k = 0,1,\dots,\dimensionof{\feature}}$ for $f$ based on $(\tilde{x},\tilde{y})$, and it is given by
\begin{align} \label{eq:ls-estimator-quadratic-model}
	\hat{\theta}_0 = \frac{1}{n} \sum_{i=1}^n y_i - \sum_{k=1}^{\dimensionof{\feature}} \frac{\beta_k}{2\alpha_k^2 n^2} \bigg( \sum_{i=1}^n x_{i, k} y_i \bigg)^2
	\qquad \text{and} \qquad
	\hat{\theta}_k = \frac{1}{\alpha_k n} \sum_{i=1}^n x_{i, k} y_i
\end{align}
for all $k \in \{1,\dots,\dimensionof{\feature}\}$. Additionally, the estimated model based on $(\tilde{x},\tilde{y})$ is given by
\begin{align} \label{eq:estimated-quadratic-model}
	f\big(x, \hat{\theta}_f(\tilde{x},\tilde{y})\big) 
	=
	\frac{1}{n} \sum_{i=1}^n y_i &+ \sum_{k=1}^{\dimx} \frac{1}{n} \bigg( \sum_{i=1}^n x_{i, k} y_i \bigg) \cdot x_k \notag \\
	&+ \sum_{k=1}^{\dimx} \frac{\beta_k}{2\alpha_k^2 n^2} \bigg( \sum_{i=1}^n x_{i, k} y_i \bigg)^2 \cdot  (x_k^2-1)
\end{align}
for every $x \in \mathcal{X}$.
\end{lm}  

\begin{proof}
In the entire proof, we write $S(\theta) := S_{\design,\observations}(\theta) := \sum_{i=1}^n (f(x_i,\theta)-y_i)^2$ for the sum of squared errors as a function of the model parameters. It then follows by straightforward calculation that 
\begin{align}
\frac{\partial_{\theta_0} S(\theta)}{2} &= n \theta_0 + \frac{n}{2} \sum_{k=1}^{\dimx} \beta_k \theta_k^2 - \sum_{i=1}^n y_i, 
\label{eq:ls-estimator-quadratic-model,step1.1}\\
\frac{\partial_{\theta_k} S(\theta)}{2} &= n \alpha_k^2 \theta_k - \alpha_k \sum_{i=1}^n x_{i,k}y_i + \beta_k \theta_k \frac{\partial_{\theta_0} S(\theta)}{2}
\qquad (k \in \{1,\dots,\dimensionof{\feature}\})
\label{eq:ls-estimator-quadratic-model,step1.2}
\end{align}
for all $\theta \in \R^{\dimp}$. Since the linear coefficients $\alpha_k$ are all assumed to be non-zero, it immediately follows from~\eqref{eq:ls-estimator-quadratic-model,step1.1} and~\eqref{eq:ls-estimator-quadratic-model,step1.2} that $S$ has a unique critical point $\hat{\theta} := (\hat{\theta}_k)_{k = 0,1,\dots,\dimensionof{\feature}}$ and that it is given by~\eqref{eq:ls-estimator-quadratic-model}. In order to show that this critical point $\hat{\theta}$ is actually a minimal point of $S$, we show that the Hessian matrix $\nabla_{\theta}^2 S(\hat{\theta})$ is positive definite. Indeed, by~\eqref{eq:ls-estimator-quadratic-model,step1.1} and~\eqref{eq:ls-estimator-quadratic-model,step1.2},
\begin{align}
\frac{\nabla_{\theta}^2 S(\hat{\theta})}{2n} 
=
\begin{pmatrix}
0 &            &        &		\\
  & \alpha_1^2 &             	\\
  &            & \ddots &      	\\
  &			   &        & \alpha_{\dimx}^2
\end{pmatrix} 
+
\begin{pmatrix}
c_0^2	& c_0 c_1 	  & c_0 c_2 	     & \dots  & c_0 c_{\dimx}     \\
c_0 c_1 & c_1^2   & c_1 c_2   & \dots & c_1 c_{\dimx} \\
c_0 c_2 & c_2 c_1 & c_2^2     & \dots & c_2 c_{\dimx} \\
\vdots & \vdots & \vdots    & \ddots & \vdots \\
c_0 c_{\dimx}  & c_{\dimx} c_1 & c_{\dimx} c_2  & \dots & c_{\dimx}^2
\end{pmatrix},
\end{align} 
where $c_0 := 1$ and $c_k := \beta_k \hat{\theta}_k$ for $k \in \{1,\dots,\dimx\}$. Consequently, 
\begin{align} \label{eq:ls-estimator-quadratic-model,step2}
v^T \frac{\nabla_{\theta}^2 S(\hat{\theta})}{2n} v = \sum_{k=1}^{\dimx} \alpha_k^2 v_k^2 + \inb{\sum_{k=0}^{\dimx} c_k v_k}^2 
\end{align}
for all $v \in \R^{\dimx+1}$. Since $\alpha_k \ne 0$ for all $k \in \{1,\dots,\dimx\}$ and $c_0 = 1 \ne 0$, it follows from~\eqref{eq:ls-estimator-quadratic-model,step2} that $v^T \nabla_{\theta}^2 S(\hat{\theta}) v > 0$ for all $v \in \R^{\dimx+1} \setminus \{0\}$. And therefore, the unique critical point $\hat{\theta}$ is indeed a minimal point of $S = S_{\design, \observations}$ and hence a least-squares estimator for $f$ based on $(\design, \observations)$. 
It remains to establish the explicit formula~\eqref{eq:estimated-quadratic-model} for the estimated model, but this immediately follows by inserting~\eqref{eq:ls-estimator-quadratic-model} into the model definition~\eqref{eq:multivariatetoymodel}. 
\end{proof}

With the formula~\eqref{eq:estimated-quadratic-model}, we can establish closed-form expressions for the expected prediction~\eqref{eq:expected-prediction} and for the prediction uncertainty~\eqref{eq:prediction-uncertainty} in a straightforward -- albeit tedious -- manner.

\begin{cor} \label{cor:exact-prediction-uncertainty-toy-example-multivariate}
Suppose $f$ is the generic quadratic model defined by~\eqref{eq:multivariatetoymodel} and $\tilde{x} = (x_1,\dots,x_n) \in \mathcal{X}$ is an experimental design satisfying~\eqref{eq:multivariatetoydesigncondition1} and \eqref{eq:multivariatetoydesigncondition2}. Suppose further that $\observations(\design)$ is the random observation model~\eqref{eq:random-observation-model} with arbitrary $\true{\parameter} \in \Theta$ and $\sigma > 0$. We then have the following formulas for the expected prediction $\mu_{\design}(x)$ and for the prediction uncertainty $V_{\design}(x)$ of the model based on $(\tilde{x}, \observations(\design))$:
\begin{align}
\mu_{\design}(x) &= f(x, \true{\theta}) + \frac{\sigma^2}{2n} \sum_{k=1}^{\dimensionof{\feature}} \frac{\beta_k}{\alpha_k^2}(x_k^2-1), 
\label{eq:exact-expected-prediction-quadratic-model}\\
V_{\design}(x) &= \frac{\sigma^2}{n} \bigg( 1 + \sum_{k=1}^{\dimensionof{\feature}} \big(x_k + \frac{\beta_k}{\alpha_k}(x_k^2-1)\true{\theta}_k \big)^2 \bigg) + \frac{\sigma^4}{2n^2} \sum_{k=1}^{\dimensionof{\feature}} \frac{\beta_k^2}{\alpha_k^4} (x_k^2-1)^2
\label{eq:exact-prediction-uncertainty-quadratic-model}
\end{align}
for every input point $x \in \mathcal{X}$. In particular, the predictions of the estimated model $x \mapsto \model(\feature, \estimator{\parameter}_{\model}(\design, \observations(\design)))$ are biased. 
\end{cor}

\begin{proof}
As a first -- preparatory -- step, we show that $\hat{\theta}_f(\design,\observations^*) = \theta^*$. Indeed, this immediately follows from the uniqueness of the least-squares estimator (Lemma~\ref{lm:least-squares-estimator-toy-example-multivariate}):
\begin{align} \label{eq:hat-theta-equals-true-theta}
\ol{\theta} := \hat{\theta}_f(\design,\observations^*) = \hat{\theta}_f(\design,\tilde{f}(\design,\theta^*)) = \theta^*.
\end{align}
As a second step, we establish~\eqref{eq:exact-expected-prediction-quadratic-model}. In order to do so, we insert the random variable $\tilde{y}(\design) = \tilde{y}^* + \tilde{\varepsilon}(\design) = (y_i^* + \varepsilon_i)_{i = 1, \dots, \designsize}$ into the explicit formula~\eqref{eq:estimated-quadratic-model} and take the expectation value, yielding
\begin{align} \label{eq:exact-prediction-uncertainty-step2.1}
\mu_{\design}(x) 
=
\frac{1}{n} \sum_{i=1}^n y_i^* &+ \sum_{k=1}^{\dimx} \frac{1}{n} \bigg( \sum_{i=1}^n x_{i, k} y_i^* \bigg) \cdot x_k \notag\\
&+ \sum_{k=1}^{\dimx} \frac{\beta_k}{2\alpha_k^2 n^2} \E\bigg( \sum_{i=1}^n x_{i, k} (y_i^* + \varepsilon_i)  \bigg)^2 \cdot  (x_k^2-1) 
\end{align}
for every $x \in \featurespace$. Since
\begin{align} \label{eq:exact-prediction-uncertainty-step2.2}
\E(\varepsilon_i) = 0
\qquad \text{and} \qquad
\E(\varepsilon_i \varepsilon_j) = \sigma^2 \delta_{i,j}
\end{align}
the claimed identity~\eqref{eq:exact-expected-prediction-quadratic-model} for the expected prediction follows from~\eqref{eq:exact-prediction-uncertainty-step2.1} using~(\ref{eq:multivariatetoydesigncondition1}.a), \eqref{eq:estimated-quadratic-model} and~\eqref{eq:hat-theta-equals-true-theta}. 
In particular, this identity tells us the prediction estimators $\model(\feature, \estimator{\parameter}_{\model}(\design, \observations(\design)))$ are biased because their expectation value $\mu_{\design}(x)$ is different from the transformed expectation value $\model(\feature, \estimator{\parameter}_{\model}(\design, \E(\observations(\design)))) = \model(\feature, \theta^*)$ by virtue of~\eqref{eq:exact-expected-prediction-quadratic-model} and~\eqref{eq:hat-theta-equals-true-theta}. 
As a third step, we establish~\eqref{eq:exact-prediction-uncertainty-quadratic-model}. In order to do so, we first observe from~\eqref{eq:estimated-quadratic-model} and~\eqref{eq:hat-theta-equals-true-theta} that 
\begin{align}  \label{eq:exact-prediction-uncertainty-step3.1}
&f(x,\hat{\theta}_f(\design,\tilde{y}(\design))) - f(x,\theta^*) 
= 
f(x,\hat{\theta}_f(\design,\tilde{y}(\design))) - f(x,\hat{\theta}_f(\design,\tilde{y}^*)) \notag\\
&\qquad=
\frac{1}{n} \sum_{i=1}^n \varepsilon_i + \sum_{k=1}^{\dimx} \frac{1}{n} \bigg( \sum_{i=1}^n x_{i, k} \varepsilon_i \bigg) \cdot x_k \notag\\
&\qquad\quad + \sum_{k=1}^{\dimx} \frac{\beta_k}{2\alpha_k^2 n^2} \bigg(\bigg( \sum_{i=1}^n x_{i, k} (y_i^* + \varepsilon_i)  \bigg)^2 - \bigg( \sum_{i=1}^n x_{i, k} y_i^*  \bigg)^2\bigg) \cdot  (x_k^2-1) \notag\\
&\qquad=
\frac{1}{n} \sum_{i=1}^n \varepsilon_i + \sum_{k=1}^{\dimx} \frac{1}{n} \bigg( \sum_{i=1}^n x_{i, k} \varepsilon_i \bigg) \cdot \big( x_k + \frac{\beta_k}{\alpha_k} \ol{\theta}_k (x_k^2-1) \big) \notag\\
&\qquad\quad + \sum_{k=1}^{\dimx} \frac{\beta_k}{2\alpha_k^2 n^2} \bigg( \sum_{i=1}^n x_{i, k} \varepsilon_i  \bigg)^2 (x_k^2-1).
\end{align}
Combining~\eqref{eq:exact-expected-prediction-quadratic-model} and~\eqref{eq:exact-prediction-uncertainty-step3.1}, we further observe that
\begin{align} \label{eq:exact-prediction-uncertainty-step3.2}
\inb{f(x,\hat{\theta}_f(\design,\tilde{y}(\design))) - \mu_{\design}(x)}^2 
=
(A + B(x) + C(x) - d(x))^2
\end{align}
for all $x \in \featurespace$. In this identity, we used the abbreviations
\begin{gather}
A := \frac{1}{\designsize} \sum_{i=1}^{\designsize} \varepsilon_i, \qquad B(x) := \frac{1}{\designsize} \sum_{k=1}^{\dimx} B_k b_k(x), \qquad C(x) := \frac{1}{2\designsize^2} \sum_{k=1}^{\dimx} C_k c_k(x), \label{eq:A,B(x),C(x)}\\
d(x) := \frac{\sigma^2}{2\designsize} \sum_{k=1}^{\dimx} \frac{\beta_k}{\alpha_k^2	} (x_k^2-1) = \frac{\sigma^2}{2\designsize} \sum_{k=1}^{\dimx} c_k(x), \label{eq:d(x)}
\end{gather}
where $B_k := \sum_{i=1}^{\designsize} x_{i,k}\varepsilon_i$ and $C_k := B_k^2$ whereas $b_k(x) := x_k + \frac{\beta_k}{\alpha_k} \ol{\theta}_k (x_k^2-1)$ and $c_k(x) := \frac{\beta_k}{\alpha_k^2} (x_k^2-1)$. So, taking the expectation value in~\eqref{eq:exact-prediction-uncertainty-step3.2}, we see that
\begin{align} \label{eq:exact-prediction-uncertainty-step3.3}
V_{\design}(x) 
&= 
\E\inb{(A + B(x) + C(x) - d(x))^2} \notag\\
&=
\E(A^2) + 2 \E(AB(x)) + 2 \E(AC(x)) - 2\E(A)d(x) + \E(B(x)^2) \\
&\quad+ 2 \E(B(x)C(x)) - 2 \E(B(x)) d(x) + \E(C(x)^2) - 2 \E(C(x)) d(x) + d(x)^2 \notag
\end{align}
for all $x \in \featurespace$. Since
\begin{gather}
\E(\varepsilon_i) = 0, \qquad \E(\varepsilon_i \varepsilon_j) = \sigma^2 \delta_{i,j}, \qquad \E(\varepsilon_i \varepsilon_j \varepsilon_r) = 0, \\
\E(\varepsilon_i^2 \varepsilon_j^2) = \sigma^4, \qquad \E(\varepsilon_i^4) = 3 \sigma^4
\end{gather}
for all $i,j,r \in \{1, \dots, \designsize\}$, the individual expectation values in~\eqref{eq:exact-prediction-uncertainty-step3.3} can be easily calculated using the assumptions~\eqref{eq:multivariatetoydesigncondition1} and~\eqref{eq:multivariatetoydesigncondition2}. Indeed, 
\begin{gather}
\E(A^2) = \frac{\sigma^2}{n}, \qquad \E(AB(x)) = 0, \qquad \E(AC(x)) = 0, \qquad \E(A) d(x) = 0 \label{eq:exact-prediction-uncertainty-step3.4}\\
\E(B(x)^2) = \frac{\sigma^2}{n} \sum_{k=1}^{\dimx} b_k(x)^2, \qquad \E(B(x)C(x)) = 0, \qquad \E(B(x))d(x) = 0, \label{eq:exact-prediction-uncertainty-step3.5}\\
\E(C(x)^2) = \frac{\sigma^4}{4\designsize^2} \bigg( \bigg(\sum_{k=1}^{\dimx} c_k(x)\bigg)^2 + 2 \sum_{k=1}^{\dimx} c_k(x)^2 \bigg), \label{eq:exact-prediction-uncertainty-step3.6}\\
\E(C(x)) d(x) = \frac{\sigma^2}{2\designsize} \sum_{k=1}^{\designsize} c_k(x) d(x), \label{eq:exact-prediction-uncertainty-step3.7}
\end{gather}
where for~\eqref{eq:exact-prediction-uncertainty-step3.6} we used that
\begin{align}
\E(C_k C_l) 
&= 
\sum_{i,j,r,s = 1}^{\designsize} x_{i,k} x_{j,k} x_{r,l} x_{s,l} \E(\varepsilon_i \varepsilon_j \varepsilon_r \varepsilon_s)  \notag\\
&=
\bigg(\sum_{i=j \ne r=s} + \sum_{i=r \ne j=s} + \sum_{i=s \ne j=r} + \sum_{i=j=r=s}\bigg) x_{i,k} x_{j,k} x_{r,l} x_{s,l} \E(\varepsilon_i \varepsilon_j \varepsilon_r \varepsilon_s) \notag\\
&=
\sigma^4 \big( \designsize^2 + 2\designsize^2 \delta_{k,l}\big).
\end{align}
Inserting now~\eqref{eq:exact-prediction-uncertainty-step3.4} to \eqref{eq:exact-prediction-uncertainty-step3.7} into~\eqref{eq:exact-prediction-uncertainty-step3.3}, we finally obtain the claimed identity~\eqref{eq:exact-prediction-uncertainty-quadratic-model} for the prediction uncertainty.
\end{proof}

\subsection{Closed-form expressions for the approximations to the prediction uncertainty}

After having established a closed-form expression for the exact prediction uncertainty of the quadratic model~\eqref{eq:multivariatetoymodel}, we now turn to the various approximations discussed in this paper.  

\begin{lm} \label{lm:least-squares-estimator-linearized-toy-example-multivariate}
Suppose $f$ is the generic quadratic model defined by~\eqref{eq:multivariatetoymodel} and $\tilde{x} = (x_1,\dots,x_n) \in \mathcal{X}$ is an experimental design satisfying~\eqref{eq:multivariatetoydesigncondition1} and \eqref{eq:multivariatetoydesigncondition2} and $\ol{\theta} = (\ol{\theta}_k)_{k=0,1,\dots,\dimx}$ is any reference parameter. Then for every $\tilde{y} = (y_1,\dots,y_n) \in \Realspace^n$, there exists a unique least-squares estimator $\hat{\theta}_{f^{\mathrm{LIN}}_{\ol{\theta}}}(\tilde{x},\tilde{y}) = (\hat{\theta}_k)_{k = 0,1,\dots,\dimx}$ for $f^{\mathrm{LIN}}_{\ol{\theta}}$ based on $(\tilde{x},\tilde{y})$, and it is given by
\begin{align} \label{eq:ls-estimator-linearized-quadratic-model}
	\hat{\theta}_0 = \frac{1}{n} \sum_{i=1}^n z_i - \sum_{k=1}^{\dimensionof{\feature}} \frac{\beta_k}{\alpha_k n} \ol{\theta}_k \bigg( \sum_{i=1}^n x_{i, k} z_i \bigg)^2
	\qquad \text{and} \qquad
	\hat{\theta}_k = \frac{1}{\alpha_k n} \sum_{i=1}^n x_{i, k} z_i
\end{align}
for all $k \in \{1,\dots,\dimensionof{\feature}\}$, where $z_i := y_i - c_{\ol{\theta}}(x_i)$ with 
\begin{align} \label{eq:c_ol-theta}
c_{\ol{\theta}}(x) := f(x,\ol{\theta}) - \nabla_{\theta} f(x,\ol{\theta})^T \ol{\theta}.
\end{align}
Additionally, the estimated linearized model based on $(\tilde{x},\tilde{y})$ is given by
\begin{align} \label{eq:estimated-linearized-quadratic-model}
	&f^{\mathrm{LIN}}_{\ol{\theta}}\big(x, \hat{\theta}_{f^{\mathrm{LIN}}_{\ol{\theta}}}(\tilde{x},\tilde{y})\big) \notag\\
	&\qquad =
	c_{\ol{\theta}}(x) + \frac{1}{n} \sum_{i=1}^n z_i + \sum_{k=1}^{\dimx} \frac{1}{n} \big( x_k + \frac{\beta_k}{\alpha_k}(x_k^2-1)\ol{\theta}_k \big) \bigg( \sum_{i=1}^n x_{i, k} z_i \bigg)
\end{align}
for every $x \in \mathcal{X}$.
\end{lm}  

\begin{proof}
In the entire proof, we write $S^{\mathrm{LIN}}_{\ol{\theta}}(\theta) := \sum_{i=1}^n (f^{\mathrm{LIN}}_{\ol{\theta}}(x_i,\theta)-y_i)^2$ for the sum of squared errors as a function of the model parameters and we use that $f^{\mathrm{LIN}}_{\ol{\theta}}(x,\theta) = c_{\ol{\theta}}(x) + \nabla_{\theta} f(x,\ol{\theta})^T \theta$ by virtue of~\eqref{eq:linearizedmodel} and~\eqref{eq:c_ol-theta}. It then follows by straightforward calculation that 
\begin{align}
\frac{\partial_{\theta_0} S^{\mathrm{LIN}}_{\ol{\theta}}(\theta)}{2} &= n \theta_0 + n \sum_{k=1}^{\dimx} \beta_k \ol{\theta}_k \theta_k - \sum_{i=1}^n z_i, 
\label{eq:ls-estimator-linearized-quadratic-model,step1.1}\\
\frac{\partial_{\theta_k} S^{\mathrm{LIN}}_{\ol{\theta}}(\theta)}{2} &= n \alpha_k^2 \theta_k - \alpha_k \sum_{i=1}^n x_{i,k}z_i + \beta_k \ol{\theta}_k \frac{\partial_{\theta_0} S^{\mathrm{LIN}}_{\ol{\theta}}(\theta)}{2}
\label{eq:ls-estimator-linearized-quadratic-model,step1.2}
\end{align}
for all $k \in \{1,\dots,\dimensionof{\feature}\}$ and $\theta \in \R^{\dimp}$. And from these relations, in turn, the assertions follow completely analogously to the proof of Lemma~\ref{lm:least-squares-estimator-toy-example-multivariate}.
\end{proof}

\begin{prop} \label{prop:linearized-prediction-uncertainty-toy-example-multivariate}
Suppose $f$ is the generic quadratic model defined by~\eqref{eq:multivariatetoymodel} and $\tilde{x} = (x_1,\dots,x_n) \in \mathcal{X}$ is an experimental design satisfying~\eqref{eq:multivariatetoydesigncondition1} and \eqref{eq:multivariatetoydesigncondition2}. Suppose further that $\observations(\design)$ is the random observation model~\eqref{eq:random-observation-model} with arbitrary $\true{\parameter} \in \Theta$ and $\sigma > 0$, that $\observations = (y_1,\dots,y_{\designsize}) \in \R^{\designsize}$ are arbitrary observations, and that 
\begin{align}
\meanof{\theta} = (\meanof{\theta}_k)_{k = 0,1,\dots,\dimensionof{\feature}} := \estimator{\parameter}_{\model}(\design, \observations) \in \Realspace^{\dimensionof{\feature}+1}
\end{align}
is the corresponding least-squares estimator. We then have the following formula for the linearization approximation $V_{\design, \observations}^{\mathrm{LIN}}(x)$ to the prediction uncertainty of $f$ based on $(\tilde{x}, \observations(\design))$:
\begin{align} \label{eq:linearized-prediction-uncertainty}
V_{\design, \observations}^{\mathrm{LIN}}(x) = \frac{\sigma^2}{n} \bigg( 1 + \sum_{k=1}^{\dimensionof{\feature}} \big(x_k + \frac{\beta_k}{\alpha_k}(x_k^2-1)\meanof{\theta}_k \big)^2 \bigg)
\qquad (x\in\mathcal{X}).
\end{align}
\end{prop}

\begin{proof}
Inserting the random variable $\tilde{y}(\design) = \tilde{y}^* + \tilde{\varepsilon}(\design) = (y_i^* + \varepsilon_i)_{i = 1, \dots, \designsize}$ into the explicit formula~\eqref{eq:estimated-linearized-quadratic-model} and taking the expectation value, we obtain
\begin{align}
f^{\mathrm{LIN}}_{\ol{\theta}}\big(x, \hat{\theta}_{f^{\mathrm{LIN}}_{\ol{\theta}}}(\design,\observations(\design))\big) - \E\big( f^{\mathrm{LIN}}_{\ol{\theta}}\big(x, \hat{\theta}_{f^{\mathrm{LIN}}_{\ol{\theta}}}(\design,\observations(\design))\big) \big)
=
A + B(x)
\end{align}
for every $x \in \featurespace$, where $A$ and $B(x)$ are defined as in~\eqref{eq:A,B(x),C(x)}. So, by~\eqref{eq:linearizationvariance}, we have
\begin{align}
V_{\design, \observations}^{\mathrm{LIN}}(x) = \E\inb{(A+B(x))^2} = \E(A^2) + 2\E(AB(x)) + \E(B(x)^2)
\end{align}
and this by~\eqref{eq:exact-prediction-uncertainty-step2.2}, in turn, yields the claimed identity~\eqref{eq:linearized-prediction-uncertainty}.
\end{proof}

\begin{prop} \label{prop:sigma-point-prediction-uncertainty-toy-example-multivariate}
Suppose $f$ is the generic quadratic model defined by~\eqref{eq:multivariatetoymodel} and $\tilde{x} = (x_1,\dots,x_n) \in \mathcal{X}$ is an experimental design satisfying~\eqref{eq:multivariatetoydesigncondition1} and \eqref{eq:multivariatetoydesigncondition2}. Suppose further that $\observations(\design)$ is the random observation model~\eqref{eq:random-observation-model} with arbitrary $\true{\parameter} \in \Theta$ and $\sigma > 0$, that $\observations = (y_1,\dots,y_{\designsize}) \in \R^{\designsize}$ are arbitrary observations, and that 
\begin{align} \label{eq:estimated-predictions}
\vectorized{\meanof{\target}} := \vectorized{\model}(\design, \meanof{\parameter})
\qquad \text{with} \qquad
\meanof{\parameter} = (\meanof{\theta}_k)_{k = 0,1,\dots,\dimensionof{\feature}} := \estimator{\parameter}_{\model}(\design, \observations)
\end{align}
are the predictions of the corresponding estimated model. Suppose finally that $\kappa \in (-n,\infty)$ is an arbitrary perturbation parameter. We then have the following formula for the sigma-point approximation $V_{\design, \observations}^{\mathrm{SP}}(x)$ to the prediction uncertainty of $f$ based on $(\tilde{x}, \observations(\design))$:
\begin{align} \label{eq:sigma-point-prediction-uncertainty-quadratic-model}
V_{\design, \observations}^{\mathrm{SP}}(x) 
&= 
\frac{\sigma^2}{n} \bigg( 1 + \sum_{k=1}^{\dimx} \big(x_k + \frac{\beta_k}{\alpha_k}(x_k^2-1)\meanof{\theta}_k \big)^2 \bigg) \notag \\
&\quad + \frac{\kappa}{n} \frac{\sigma^4}{4n^2} \bigg( \sum_{k=1}^{\dimx} \frac{\beta_k}{\alpha_k^2} (x_k^2-1) \bigg)^2
\qquad (x\in\mathcal{X}).
\end{align}
\end{prop}

\begin{proof}
As a first step, we show that $\hat{\theta}_f(\design,\tilde{\ol{y}}) = \ol{\theta}$. Indeed, this immediately follows from the definition~\eqref{eq:estimated-predictions} and from the uniqueness of the least-squares estimator (Lemma~\ref{lm:least-squares-estimator-toy-example-multivariate}):
\begin{align} \label{eq:ls-of-estimated-observations-equals-ls}
\hat{\theta}_f(\design,\tilde{\ol{y}}) = \hat{\theta}_f(\design,\tilde{f}(\design,\ol{\theta})) = \ol{\theta}.
\end{align}
As a second step, we establish the claimed identity~\eqref{eq:sigma-point-prediction-uncertainty-quadratic-model}. In order to do so, we insert our sigma points $\tilde{\ol{y}} \pm \delta e_i$ into the explicit formula~\eqref{eq:estimated-quadratic-model}. In this manner, we obtain
\begin{align} \label{eq:sigma-point-prediction-uncertainty,step2.1}
&g_{\design,\observations}(x,\pm \delta \cdot e_i) = f\inb{x, \hat{\theta}_f(\design, \tilde{\ol{y}} \pm \delta e_i)} \\
&=
f(x,\ol{\theta}) \pm \frac{\delta}{\designsize} \bigg( 1 + \sum_{k=1}^{\dimx} x_{i,k} x_k + \sum_{k=1}^{\dimx} \frac{\beta_k}{\alpha_k} \ol{\theta}_k x_{i,k} (x_k^2-1) \bigg) + \frac{\delta^2}{2\designsize^2} \sum_{k=1}^{\dimx} \frac{\beta_k}{\alpha_k^2} (x_k^2-1) \notag
\end{align}
for every $x \in \featurespace$, where aside from~\eqref{eq:estimated-quadratic-model} we also used~\eqref{eq:ls-of-estimated-observations-equals-ls}. Inserting~\eqref{eq:sigma-point-prediction-uncertainty,step2.1} into the definition~\eqref{eq:sigma-point-expected-prediction} of $\mean_{\design, \observations}^{\mathrm{SP}}$, we conclude that
\begin{align} \label{eq:sigma-point-prediction-uncertainty,step2.2}
\mean_{\design, \observations}^{\mathrm{SP}}(x) = f(x,\ol{\theta}) + \frac{\sigma^2}{2\designsize} \sum_{k=1}^{\dimx} \frac{\beta_k}{\alpha_k} (x_k^2-1).
\end{align}
Inserting~\eqref{eq:sigma-point-prediction-uncertainty,step2.1} and~\eqref{eq:sigma-point-prediction-uncertainty,step2.2} into the definition~\eqref{eq:squared-error-sigma-points} of $h_{\design,\observations}^{\mathrm{SP}}$, we further conclude that
\begin{align} \label{eq:sigma-point-prediction-uncertainty,step2.3}
h_{\design,\observations}^{\mathrm{SP}}(x,\pm \delta \cdot e_i) 
&= 
\inb{f\inb{x, \hat{\theta}_f(\design, \tilde{\ol{y}} \pm \delta e_i)} - \mean_{\design, \observations}^{\mathrm{SP}}(x)}^2 \notag\\
=
&\inb{\varphi_{\pm}(x) \pm \psi_i(x)}^2 
\end{align}
for every $x \in \featurespace$, where we used the abbreviations
\begin{gather}
\varphi_{\pm}(x) := \frac{\delta}{\designsize} + \frac{\kappa \sigma^2}{2\designsize^2} \sum_{k=1}^{\dimx} c_k(x), \qquad
\psi_i(x) := \frac{\delta}{\designsize} \sum_{k=1}^{\dimx} b_k(x) x_{i,k}
\end{gather}
with $b_k(x), c_k(x)$ as in~\eqref{eq:A,B(x),C(x)} -- in conjunction with the fact that $\frac{\delta^2}{2\designsize^2} - \frac{\sigma^2}{2\designsize} = \frac{\kappa \sigma^2}{2\designsize^2}$ by virtue of~\eqref{eq:increment-sigma-points}. Combining now~\eqref{eq:sigma-point-prediction-uncertainty,step2.3} with \eqref{eq:multivariatetoydesigncondition1} and \eqref{eq:multivariatetoydesigncondition2}, we obtain
\begin{align}
&\sum_{i=1}^{\designsize} \big(h_{\design,\observations}^{\mathrm{SP}}(x,\delta \cdot e_i) + h_{\design,\observations}^{\mathrm{SP}}(x,-\delta \cdot e_i)\big) 
= 
n (\varphi_-(x)^2 + \varphi_+(x)^2) + 2 \sum_{i=1}^{\designsize} \psi_i(x)^2 \notag\\
&=
\frac{2\delta^2}{\designsize} \bigg(1 + \sum_{k=1}^{\dimx} b_k(x)^2 \bigg) + \frac{\kappa^2 \sigma^4}{2 \designsize^3} \bigg( \sum_{k=1}^{\dimx} c_k(x) \bigg)^2.
\end{align}
And from this, in turn, the claimed identity~\eqref{eq:sigma-point-prediction-uncertainty-quadratic-model} follows in a straightforward manner using the definition~\eqref{eq:sigma-point-prediction-uncertainty}.
\end{proof}

\begin{cor} \label{cor:comparison-of-prediction-uncertainty-measures-toy-example-multivariate}
Suppose $f$ is the generic quadratic model defined by~\eqref{eq:multivariatetoymodel} and $\tilde{x} = (x_1,\dots,x_n) \in \mathcal{X}$ is an experimental design satisfying~\eqref{eq:multivariatetoydesigncondition1} and \eqref{eq:multivariatetoydesigncondition2}. Suppose further that $\observations(\design)$ is the random observation model~\eqref{eq:random-observation-model} with arbitrary $\true{\parameter} \in \Theta$ and $\sigma > 0$, and that the observations 
\begin{align} \label{eq:noisefree-observations}
\observations = \true{\observations} := \vectorized{\model}(\design, \true{\parameter})
\end{align}
are exactly equal to the predictions of the true model at $\design$ (noise-free observations). Then the linearization and the sigma-point approximation to the prediction uncertainty are not exact, namely 
\begin{align} \label{eq:sigma-point-uncertainty-between-linearized-and-exact-uncertainty-multivariate}
V_{\design, \observations}^{\mathrm{LIN}}(x) < V_{\design}(x) < V_{\design, \observations}^{\mathrm{SP}}(x)
\qquad (x\in\mathcal{X} \setminus \{-1,1\}^{\dimensionof{\feature}})
\end{align}
for all sigma-point perturbation parameters $\kappa \in (2n,\infty)$ provided that the quadratic model coefficients $\beta_k$ are all non-zero with the same sign:
\begin{align} \label{eq:quadratic-model-coefficients-have-the-same-sign}
\beta_k < 0 \qquad (k \in \{1,\dots,\dimx\})
\qquad \text{or} \qquad
\beta_k > 0 \qquad (k \in \{1,\dots,\dimx\}). 
\end{align}
In contrast, the McNamee-Stenger and the Lu-Darmofal approximations are exact:
\begin{align} \label{eq:mcnamee-stenger-and-lu-darmofal-are-exact}
V_{\design, \observations}^{\mathrm{MS}}(x) = V_{\design}(x) = V_{\design, \observations}^{\mathrm{LD}}(x)
\qquad (x \in \mathcal{X}).
\end{align}
\end{cor}

\begin{proof}
As a first -- preparatory -- step, we show that $\ol{\theta} = \theta^*$. Indeed, this immediately follows from the noise-free observation assumption~\eqref{eq:noisefree-observations} and the uniqueness of the least-squares estimator (Lemma~\ref{lm:least-squares-estimator-toy-example-multivariate}):
\begin{align} \label{eq:ol-theta-equals-true-theta}
\ol{\theta} := \hat{\theta}_f(\design,\observations) = \hat{\theta}_f(\design,\tilde{f}(\design,\theta^*)) = \theta^*.
\end{align}
As a second step, we establish the inexactness relation~\eqref{eq:sigma-point-uncertainty-between-linearized-and-exact-uncertainty-multivariate}. Indeed, let 
\begin{align} \label{eq:kappa-larger-than-2n}
\kappa \in (2n,\infty)
\end{align}
and let~\eqref{eq:quadratic-model-coefficients-have-the-same-sign} be satisfied. It then follows by Corollary~\ref{cor:exact-prediction-uncertainty-toy-example-multivariate} and Propositions~\ref{prop:linearized-prediction-uncertainty-toy-example-multivariate} and~\ref{prop:sigma-point-prediction-uncertainty-toy-example-multivariate} in conjunction with~\eqref{eq:ol-theta-equals-true-theta} that
\begin{align}
V_{\design}(x) 
&= V_{\design,\observations}^{\mathrm{LIN}}(x) + \frac{\sigma^2}{2\designsize^2} \sum_{k=1}^{\dimx} \frac{\beta_k^2}{\alpha_k^4} (x_k^2-1)^2, 
\label{eq:ms-and-ld-exact,step2.1}\\
V_{\design,\observations}^{\mathrm{SP}}(x) 
&= V_{\design}(x) + \frac{\sigma^2}{2\designsize^2} \frac{\kappa}{2\designsize} \bigg( \sum_{k=1}^{\dimx} \frac{\beta_k}{\alpha_k^2} (x_k^2-1) \bigg)^2 
- \frac{\sigma^2}{2\designsize^2} \sum_{k=1}^{\dimx} \frac{\beta_k^2}{\alpha_k^4} (x_k^2-1)^2 
\label{eq:ms-and-ld-exact,step2.2}
\end{align}
for all $x \in \featurespace$. In view of~\eqref{eq:quadratic-model-coefficients-have-the-same-sign}, it further follows that 
\begin{align} \label{eq:ms-and-ld-exact,step2.3}
\sum_{k=1}^{\dimx} \frac{\beta_k^2}{\alpha_k^4} (x_k^2-1)^2 > 0
\qquad (x\in\mathcal{X} \setminus \{-1,1\}^{\dimensionof{\feature}})
\end{align}
and therefore by~\eqref{eq:kappa-larger-than-2n}
\begin{align}  \label{eq:ms-and-ld-exact,step2.4}
\frac{\kappa}{2\designsize} \bigg( \sum_{k=1}^{\dimx} \frac{\beta_k}{\alpha_k^2} (x_k^2-1) \bigg)^2  
&= \frac{\kappa}{2\designsize} \bigg(\sum_{k=1}^{\dimx} \frac{\beta_k^2}{\alpha_k^4} (x_k^2-1)^2 + \sum_{k \ne k} \frac{\beta_k}{\alpha_k^2} \frac{\beta_k}{\alpha_k^2} \cdot (x_k^2-1) (x_k^2-1)\bigg) \notag\\
&> 
\sum_{k=1}^{\dimx} \frac{\beta_k^2}{\alpha_k^4} (x_k^2-1)^2 
\qquad (x\in\mathcal{X} \setminus \{-1,1\}^{\dimensionof{\feature}}).
\end{align}
Combining these strict inequalities with~\eqref{eq:ms-and-ld-exact,step2.1} and~\eqref{eq:ms-and-ld-exact,step2.2}, we obtain the claimed inexactness relations~\eqref{eq:sigma-point-uncertainty-between-linearized-and-exact-uncertainty-multivariate}. 
As a third step, we establish the exactness relation~\eqref{eq:mcnamee-stenger-and-lu-darmofal-are-exact}. Indeed, by~\eqref{eq:estimated-quadratic-model}, the functions
\begin{align}
\tilde{z} \mapsto g_{\design,\observations}(x, \tilde{z}) &:= f(x, \hat{\theta}_f(\design,\tilde{\ol{y}} + \tilde{z})) 
\label{eq:ms-and-ld-exact,step3.1}\\
\tilde{z} \mapsto h_{\design,\observations}^{\cubature}(x, \tilde{z}) &:= \inb{ f(x, \hat{\theta}_f(\design,\tilde{\ol{y}} + \tilde{z})) - \mu_{\design,\observations}^{\cubature}(x)}^2
\label{eq:ms-and-ld-exact,step3.2}
\end{align}
defined in~\eqref{eq:model-trained-on-perturbed-observations} and in~\eqref{eq:squared-error-mcnamee-stenger} and~\eqref{eq:squared-error-lu-darmofal} are polynomial functions of degree at most $2$ or $4$, respectively. So, as the cubature formulas of McNamee-Stenger and of Lu-Darmofal are both exact up to fifth degree, we conclude the exact identities
\begin{align} \label{eq:ms-and-ld-exact,step3.3}
\int_{\R^{\designsize}} g_{\design,\observations}(x, \tilde{z}) p(\tilde{z}) \d \tilde{z} = \mu_{\design,\observations}^{\cubature}(x), 
\qquad
\int_{\R^{\designsize}} h_{\design,\observations}^{\cubature}(x, \tilde{z}) p(\tilde{z}) \d \tilde{z} = V_{\design,\observations}^{\cubature}(x)
\end{align}
for every $x \in \featurespace$ and for $\cubature \in \{\mathrm{MS}, \mathrm{LD}\}$. In view of~\eqref{eq:ol-theta-equals-true-theta}, we further conclude that
\begin{align}
\tilde{\ol{y}} = \tilde{f}(\design,\ol{\theta}) = \tilde{f}(\design,\theta^*) = \tilde{y}^*
\end{align}
and therefore we see, by~\eqref{eq:ms-and-ld-exact,step3.1}, \eqref{eq:ms-and-ld-exact,step3.2} and~\eqref{eq:ms-and-ld-exact,step3.3}, that
\begin{align}
\int_{\R^{\designsize}} g_{\design,\observations}(x, \tilde{z}) p(\tilde{z}) \d \tilde{z} 
&= \int_{\R^{\designsize}} f(x, \hat{\theta}_f(\design,\tilde{y}^* + \tilde{z})) p(\tilde{z}) \d \tilde{z}
= \mu_{\design,\observations^*}(x), 
\label{eq:ms-and-ld-exact,step3.4}\\
\int_{\R^{\designsize}} h_{\design,\observations}^{\cubature}(x, \tilde{z}) p(\tilde{z}) \d \tilde{z} 
&= \int_{\R^{\designsize}} \inb{ f(x, \hat{\theta}_f(\design,\tilde{y}^* + \tilde{z})) - \mu_{\design,\observations}(x)}^2 p(\tilde{z}) \d \tilde{z} \notag\\
&= V_{\design}(x)
\label{eq:ms-and-ld-exact,step3.5}
\end{align}
for every $x \in \featurespace$ and for $\cubature \in \{\mathrm{MS}, \mathrm{LD}\}$. Combining now~\eqref{eq:ms-and-ld-exact,step3.3} and \eqref{eq:ms-and-ld-exact,step3.5}, we finally obtain the claimed exactness relations~\eqref{eq:mcnamee-stenger-and-lu-darmofal-are-exact}.
\end{proof}

\section*{Acknowledgments}
We gratefully acknowledge funding from the Deutsche Forschungsgemeinschaft (DFG, German Research Foundation) – project number 466397921 – within the Priority Programme ``SPP 2331: Machine Learning in Chemical Engineering''.

\printbibliography

\end{document}